\documentclass[journal]{IEEEtran}
\usepackage[utf8]{inputenc}
\usepackage{lipsum}
\usepackage{graphicx}
\ifCLASSOPTIONcompsoc
    \usepackage[caption=false, font=normalsize, labelfont=sf, textfont=sf]{subfig}
\else
\usepackage[caption=false, font=footnotesize]{subfig}
\fi

\usepackage{cite}
\usepackage{amsmath,amsthm,amssymb}
\usepackage{epstopdf}
\usepackage{threeparttable}
\usepackage{multirow}
\usepackage{hyperref}
\usepackage{caption}
\usepackage{mathtools}
\usepackage{lscape}
\newcounter{MYtempeqncnt}
\usepackage{float}
\usepackage{stfloats}
\usepackage{booktabs} 
\usepackage{color}

\newtheorem{theorem}{Theorem}

\newtheorem{proposition}{Proposition}

\newtheorem{remark}{Remark}

\usepackage{tabularx}
\usepackage{rotating}  
\usepackage[T1]{fontenc}
\DeclareMathOperator*{\argmax}{argmax}
\DeclareMathOperator*{\argmin}{argmin}
\usepackage{arydshln} 

\ifCLASSINFOpdf
\else
\fi
\usepackage[cmintegrals]{newtxmath}
\begin{document}
\title{On Physical Layer Security over Fox's $H$-Function Wiretap Fading Channels}
\author{Long~Kong,~Georges~Kaddoum,~\IEEEmembership{Member,~IEEE}, and Hatim~Chergui,~\IEEEmembership{Member,~IEEE}
\thanks{L. Kong and G. Kaddoum are with the LaCIME lab, Department
of Electrical Engineering, École de technologie supérieure (ÉTS), Université du Québec, Montréal (Québec), Montreal,
Canada, H3C 1K3, e-mails: (long.kong.1@ens.estmtl.ca, georges.kaddoum@etsmtl.ca.}
\thanks{H. Chergui is with CTTC, Barcelona, Spain.
e-mail: chergui@ieee.org.}
\thanks{Manuscript received xxxxxx.}}


\maketitle

\begin{abstract}
Most of the well-known fading distributions, if not all of them, could be encompassed by the Fox's $H$-function fading. Consequently, we investigate the physical layer security (PLS) over Fox's $H$-function fading wiretap channels, in the presence of non-colluding and colluding eavesdroppers. In particular, for the non-colluding scenario, closed-form expressions are derived for the secrecy outage probability (SOP), the probability of non-zero secrecy capacity (PNZ), and the average secrecy capacity (ASC). These expressions are given in terms of either univariate or bivariate Fox's $H$-function. In order to show the effectiveness of our derivations, three metrics are respectively listed over the following frequently used fading channels, including Rayleigh, Weibull, Nakagami-$m$, $\alpha-\mu$, Fisher-Snedecor (F-S) $\mathcal{F}$, and extended generalized-$\mathcal{K}$ (EGK). Our tractable results are not only straightforward and general, but also feasible and applicable, especially the SOP, which is usually limited to the lower bound in the literature due to the difficulty of deriving closed-from analytical expressions. For the colluding scenario, a super eavesdropper equipped with maximal ratio combining (MRC) or selection-combining (SC) schemes is characterized. The lower bound of SOP and exact PNZ are thereafter derived with closed-form expressions in terms of the multivariate Fox's $H$-function. In order to validate the accuracy of our analytical results, Monte-Carlo simulations are subsequently conducted for the aforementioned fading channels. One can observe that for the former non-colluding scenario, we have perfect agreement between the exact analytical and simulation results, and highly accurate approximations between the exact and asymptotic analytical results. On the contrary, the SOP and PNZ of colluding eavesdropper is greatly degraded with the increase of the number of eavesdroppers. Also, the so-called super eavesdropper with MRC is much powerful to wiretap the main channel than the one with SC. 
\end{abstract}
\IEEEpeerreviewmaketitle

\section{Introduction} 
\IEEEPARstart{D}{ifferent} wireless systems are usually characterized with various statistical models. For example, the gamma-gamma distribution was introduced to model the free space optical (FSO) communication link \cite{8358703,7968415}, and Fisher-Snedecor (F-S) $\mathcal{F}$ to model the device-to-device communication \cite{7886273,8359199}. As such, many endeavors have been drawn to investigate the mathematical characteristics of secure transmission for different communication scenarios. 

Dating back to the fundamental works of physical layer security (PLS) from the information theoretical perspective, Shannon and Wyner are undoubtedly the pioneers in this field \cite{Shan49,6772207}. They established the mathematical background of perfect secrecy and wiretap channel models. Later on, Wyner's classic wiretap model was investigated over additive white Gaussian noise channel (AWGN) and Rayleigh fading channels \cite{Leung78,4529264}. Over the past decades, plenty of research efforts have been pursued on the investigation of PLS over various fading channels, such as Rayleigh \cite{4529264}, Rician \cite{6338984,IETletter2}, Nakagami-$m$, Weibull \cite{6831147}, Lognormal \cite{7058442}, generalized-$\mathcal{K}$ \cite{7313027,7342908,7654915,8279403,8672771}, and $\alpha-\mu$ (or, equivalently, generalized gamma) \cite{7094262,7374839,7856980,ICCWORKSHOPlong2018,OutageLongLett,8477185}, $\alpha-\eta-\kappa-\mu$ \cite{8421203}, etc. Secrecy outage probability (SOP), the probability of non-zero secrecy capacity (PNZ), and the average secrecy capacity (ASC) are the three typical and frequently studied secrecy metrics.

As more new communication topologies appear, e.g., device-to-device (D2D) communications, FSO communications, intervehicle communication, millimeterwave  (mmWave) communications, wireless body area networks (WBAN), and cognitive radios, the existing models become obsolete. As such, more advanced and better suited fading models were subsequently proposed and analyzed, such as $\alpha-\mu$ \cite{4067122}, $\kappa-\mu/\eta-\mu$ \cite{4231253}, F-S $\mathcal{F}$ \cite{7886273,8359199}, the extended generalized-$\mathcal{K}$ (EGK) \cite{6226905}, cascaded $\alpha-\mu$ fading \cite{8354927}, among many other fading channels.  

With the emergence of various fading models, a unified and generic fading model is required to subsume most, if not all, of these fading distributions. Fox's $H$-function distribution, reported in \cite{7089281,7470056,7547496}, is one possible model to accommodate various fading models with high flexibility. It was first introduced in \cite{bodenschatz1992finding} and \cite{cook1981h} as a pure mathematical finding, and can be generalized to the Gamma, exponential, Chi-square, Weibull, Rayleigh, Half-Normal distribution, etc. Other examples, including generalized-$\mathcal{K}$, $\alpha-\mu$, F-S $\mathcal{F}$, and EGK, were recently explored by Alhennawi \textit{et al.} \cite{7089281} and Rahama \textit{et al.} \cite{8338118}. These findings were achieved by transforming these probability density distributions (PDFs) of received signal-to-noise ratios (SNRs) in terms of Fox's $H$-function.

The feasibility and applicability of Fox's $H$-function distribution as a general fading model for wireless communication is not new. In \cite{6226905}, a variation of Fox's $H$-function fading model was proposed as a general model for most well-known distributions. Jeong \textit{et al.} found that Fox's $H$-function distribution offers a better fading model of vehicle-to-vehicle (V2V) communication than other ordinary fading distributions \cite{6550886}. More recently, Alhennawi \textit{et al.} in \cite{7089281} derived the symbol error rate (SER) and channel capacity of single- and multiple-branch diversity receivers when communicating over Fox's $H$-function fading channels. As a consequence, the advantages of Fox's $H$-function fading are threefold: 
\begin{itemize}
\item The genericity of its form for most distribution, e.g., Rayleigh, Nakagami-\textit{m}, Weibull, $\alpha-\mu$, etc;
\item The simplicity and the generality of it to derive the key performance metrics of wireless communications systems, e.g., outage probability, SER, and channel capacity \cite{7089281}. 
\item The possibility of using its distribution to study the PLS analysis over $\alpha-\mu$, F-S $\mathcal{F}$ fading channels \cite{7856980,OutageLongLett,LongFisherF,8354927}.
\end{itemize}
To the best of the authors' knowledge, apart from the investigation of PLS over the aforementioned fading channels \cite{4529264,6338984,IETletter2,6831147,7058442,7313027,7342908,7654915,8279403,8672771,7094262,7374839,7856980,ICCWORKSHOPlong2018,OutageLongLett}, including generalized-$\mathcal{K}$, $\alpha-\mu$, $\kappa-\mu$ \cite{7467556,7996664,8013132}, F-S $\mathcal{F}$ \cite{LongFisherF}, no works has ever been found to analyze the PLS over the general Fox's $H$-function fading channels. To this end, this paper is subject to the investigation of PLS over Fox's $H$-function fading channels, with consideration of the non-colluding and colluding eavesdropping scenarios.
\vspace{-0.5cm}
\subsection{Our Work and Contributions}
The contributions of this paper are multifold, which are listed as follows:
\begin{itemize}
\item[1)] Novel exact and closed-form expressions are initially derived for the secrecy metrics, including the SOP, PNZ, and ASC. Our formulations, in terms of univariate or bivariate Fox's $H$-function, are given in \textit{simple} and \textit{tractable} mathematical expressions.
\item[2)] The difficulty of deriving closed-form expressions for the SOP explicitly lies in tractable integrals. Consequently, many works can be found on the development of lower bound of the SOP ($\mathcal{P}_{out}=Pr(C_s \ge R_s)$). Since the lower bound of SOP is actually the complementary of the probability of non-zero secrecy capacity, i.e., $\mathcal{P}_{nz} = Pr(C_s >0)$, it is much easier to obtain the lower bound of the SOP and PNZ, which can be found in \cite{4067122}. Strictly speaking, our work fills this gap of lacking exact closed-form SOP expressions. 
\item[3)] The obtained general and unified secrecy metrics' expressions are found identical with the existing works when being compared with Monte-Carlo simulation results. Moreover, the obtained secrecy expressions can be straightforward applied to other transformable but not listed herein wiretap fading channels.
\item[4)] The asymptotic behaviors of these secrecy metrics are also obtained for the sake of providing simple but highly accurate approximations of secrecy metrics at high average signal-to-noise (SNR) regime. 
\item[5)] Considering the colluding eavesdropping scenario with maximal ratio combining (MRC) and selection combining (SC) schemes, the lower bound of the SOP and exact PNZ are characterized in terms of multivariate Fox's $H$-function.  
\end{itemize}
Resultantly, the obtained analytical results are especially beneficial since the analytical expressions themselves (i) provide a unified approach to analyze the PLS over a generalized fading model; (ii) serve as an efficient and convenient tool to validate and compare the special cases of Fox's $H$-function fading channels; and (iii) enable researchers and wireless communication engineers to quickly evaluate secrecy performance when encountering security risks.
\vspace{-0.3cm}
\subsection{Structure and Notations}
The rest of this paper is structured as follows: Section \ref{Sec_Preliminary} illustrates Fox's $H$-function fading and its Mellin transform. In Section \ref{Sec_systemmodel}, the system model and problem formulation are presented. In the presence of non-colluding and colluding scenarios, secrecy analysis are respectively conducted in Sections \ref{Ana_noncoll}, \ref{Ana_Asym}, and \ref{Label_Colluding}, together with several examples. Afterwards, in Section \ref{Sec_Num}, numerical results and discussions are presented. Finally, Section \ref{Sec_con} concludes the paper.

\textit{Mathematical Functions and Notations}: $j \triangleq \sqrt{-1}$, $\Gamma(.)$ is the complete Gamma function, $H_{p,q}^{m,n}[.]$ is the univariate Fox's $H$-function \cite[eq. (1.2)]{mathai2009h}, $H_{p,q;p_1,q_1;p_2,q_2}^{m,n;m_1,n_1;m_2,n_2}$ is the extended generalized bivariate Fox's $H$-function \cite[eq. (2.56)]{mathai2009h}. $H_{p,q;p_1,q_1;\cdots;p_L,q_L}^{m,n;m_1,n_1;\cdots;m_L,n_L}$ is the multivariate Fox's $H$-function \cite[eq. (2.56)]{mathai2009h}. $f(x)$ and $F(x)$ represent the probability density function (PDF) and cumulative distribution function (CDF) of $x$, respectively. $\mathcal{B}(x,y)$ is the Beta function \cite[eq. (8.380.1)]{gradshteyn2014table}. $\mathcal{M}[f(x),s]$ denotes the Mellin transform of $f(x)$. $\text{Res}[f(x),s]$ represents the residue of function $f(x)$ at pole $x=p$. $\Psi_0(\cdot)$ is the digamma function.

\section{Preliminary} \label{Sec_Preliminary}
\subsection{Fox's $H$-function Fading}
Consider a wireless communication link over a fading channel, where the instantaneous SNR at user $k$, $\gamma_k$, follows Fox's $H$-function PDF, given by \cite{bodenschatz1992finding}
\begin{equation} \label{PDF_foxH}
\begin{split}
f_{k}(\gamma_k)& = \kappa H_{p,q}^{m,n} \left[ {\lambda \gamma_k \left| {\begin{array}{*{20}c}
    {(a_i,A_i)_{i=1:p}}   \\
   {(b_l,B_l)_{l=1:q}}  \\
\end{array}} \right.} \right], \quad \gamma >0, \\
&\mathop =^{(a)} \frac{\kappa}{2\pi j} \int_{\mathcal{L}} \Theta_k (s) (\lambda \gamma_k)^{-s} ds,
\end{split}
\end{equation}
where $\lambda > 0$ and $\kappa$ are constants such that $\int_0^\infty f_k (\gamma_k) d \gamma_k = 1$. $(x_i, y_i)_l$ is a shorthand for $(x_1,y_1),\cdots,(x_l,y_l)$. Step $(a)$ is developed by expressing Fox's $H$-function in terms of its definition \cite[eq. (1.2)]{mathai2009h}. $A_i > 0$ for all $i=1,\cdots, p$, and $B_l > 0$ for all $l = 1, \cdots, q$. $ 0 \le m \le q$, $0 \le n \le p$, $\mathcal{L}$ is a suitable contour separating the poles of the gamma functions $\Gamma(b_l+B_l s)$ from the poles of the gamma functions $\Gamma(1 - a_i -A_i s)$,
\begin{equation}
\Theta_k (s) = \frac{\prod \limits_{l=1}^m \Gamma(b_l + B_l s) \prod \limits_{i=1}^n \Gamma(1 - a_i - A_is)}{\prod \limits_{l=m+1}^q\Gamma(1-b_l - B_l s) \prod \limits_{i=n+1}^p\Gamma(a_i + A_is)} .
\end{equation}

The cumulative distribution function (CDF) of the received SNR at user $k$, i.e., $\gamma_k$ is given by \cite[eqs. (3.9) and (3.7)]{bodenschatz1992finding}
\begin{subequations}
\begin{equation}
\begin{split}
F_k(\gamma_k)& = \frac{\kappa}{\lambda} H_{p+1,q+1}^{m,n+1} \left[ {  \lambda \gamma_k \left| {\begin{array}{*{20}c}
    {(1,1),(a_i+A_i,A_i)_p}   \\
   {(b_l+B_l,B_l)_q,(0,1)}  \\
\end{array}} \right.} \right] ,
\end{split}
\end{equation}
or
\begin{equation}
\begin{split}
F_k(\gamma_k)& = 1 - \frac{\kappa}{\lambda} H_{p+1,q+1}^{m+1,n} \left[ {  \lambda \gamma_k \left| {\begin{array}{*{20}c}
    {(a_i+A_i,A_i)_p,(1,1)}   \\
   {(0,1),(b_l+B_l,B_l)_q}  \\
\end{array}} \right.} \right] \\
& = 1 - \bar{F}_{k}(\gamma_k),
\end{split}
\end{equation}
\end{subequations}
where $\bar{F}_{k}(\gamma)$ is the complementary CDF (CCDF). For the notational convenience, $\Theta_{k}^f$ and $\Theta_{k}^F$ are used thereafter to denote the PDF and CDF of Fox's $H$-function, respectively. The Mellin transform of $f_k(\gamma)$ is defined and given as \cite[eq. (5)]{7089281} \cite[eq. (2.8)]{mathai2009h},
\begin{equation}
\mathcal{M}[ f_k(\gamma_k),s] = \int_0^\infty f_k(\gamma_k) \gamma^{s-1} d \gamma_k = \kappa \lambda^{-s} \Theta_k(s).
\end{equation}
\subsection{Special cases}
As mentioned before, Fox's $H$-function distribution provides enough flexibility to accommodate most fading distributions. As a result, the objective herein is to list some well-known examples, such as the $\alpha-\mu$\footnotemark[1]\footnotetext[1]{Since $\alpha-\mu$ distributions can be attributed to exponential, one-sided Gaussian, Rayleigh, Nakagami-$m$, Weibull and Gamma fading distributions by assigning specific values for $\alpha$ and $\mu$, respectively \cite{4067122}, secrecy analysis on these fading distributions is thus omitted herein.}, F-S $\mathcal{F}$, and EGK, as shown in Table. \ref{T1}, where $\bar{\gamma}_k$ is the average received SNR at user $k$.  
\begin{table}[!t]
\newcommand{\tabincell}[2]{\begin{tabular}{@{}#1@{}}#2\end{tabular}}
  \begin{center}
  \renewcommand{\arraystretch}{1.4}
    \caption{Exact expressions of $f_{k}(\gamma_k)$ for different special cases of Fox's $H$-function distribution} 
    \label{T1}
    \begin{tabular}{p{38pt} l}
      \toprule
       Ins. SNR & $f_k(\gamma_k)$  \\
      \midrule
  \tabincell{c}{\textbf{$\alpha-\mu$}\\ \cite[Tab. V]{6226905}} &  \tabincell{l}{$f_k(\gamma_k )= \kappa H_{0,1}^{1,0} \left[ {  \lambda \gamma_k \left| \hspace{-1ex}{\begin{array}{*{20}c}
    {-}   \\
   {(\mu-\frac{1}{\alpha},\frac{1}{\alpha})}  \\
\end{array}} \right.}\hspace{-1ex} \right],$\\ \specialrule{0em}{1pt}{1pt} where $\kappa = \frac{\beta}{\Gamma(\mu)\bar{\gamma}_k},\lambda = \frac{\beta}{\bar{\gamma}_k}, \beta = \frac{\Gamma\left(\mu + \frac{1}{\alpha}\right)}{\Gamma(\mu)}$.} \\ \specialrule{0em}{1pt}{1pt} \hline \specialrule{0em}{1pt}{1pt}
\tabincell{c}{\textbf{F-S}  $\mathcal{F}$ \\ \cite[eq. (5)]{7886273}} &  \tabincell{l}{$f_k(\gamma_k )=\kappa H_{1,1}^{1,1} \left[ {  \lambda\gamma_k \left| \hspace{-1ex} {\begin{array}{*{20}c}
    {(-m_{k,s},1)}   \\
   {(m_k-1,1)}  \\
\end{array}} \right.} \hspace{-1ex} \right],$\\  \specialrule{0em}{1pt}{1pt}
where $\kappa = \frac{\lambda }{\Gamma(m_k)\Gamma(m_{k,s})}, \lambda = \frac{m_k}{m_{k,s}\bar{\gamma}_k}$.} \\ \specialrule{0em}{1pt}{1pt} \hline \specialrule{0em}{1pt}{1pt}
      \tabincell{c}{\textbf{EGK} \\ \cite[eq. (18)]{8338118}} &  \tabincell{l}{$f_k(\gamma_k )=\kappa  H_{0,2}^{2,0} \left[ {  \lambda \gamma_k \left| \hspace{-1ex}{\begin{array}{*{20}c}
    {-}   \\
   {(m_l - \frac{1}{\xi_l},\frac{1}{\xi_l}),(m_{sl}-\frac{1}{\xi_{sl}},\frac{1}{\xi_{sl}})}  \\
\end{array}} \right.} \hspace{-1ex}\right],$ \\ \specialrule{0em}{1pt}{1pt}
where $\kappa = \frac{\beta_l \beta_{sl}}{\Gamma(m_l)\Gamma(m_{sl})\bar{\gamma}_{k}}, \lambda = \frac{\beta_l\beta_{sl}}{\bar{\gamma}_{k}}, \beta_{l} = \frac{\Gamma\left(m_l + \frac{1}{\xi_l}\right)}{\Gamma(m_l)}$,\\ \quad and $ \beta_{sl} = \frac{\Gamma\left(m_{sl} + \frac{1}{\xi_{sl}}\right)}{\Gamma(m_{sl})}$.  }\\
      \bottomrule
    \end{tabular}
  \end{center}
\end{table} 
\section{System Model and Problem formulation} \label{Sec_systemmodel}
\subsection{System Model}
The Alice-Bob-Eve classic wiretap model is used here to illustrate a legitimate transmission link (Alice $\rightarrow$ Bob) in the presence of a malicious eavesdropper. In such a wiretap channel model, the transmitter Alice (A) wishes to send secret messages to the intended receiver Bob (B) in the presence of an eavesdropper Eve (E); the link between A and B is called the main channel, whereas the one between A and E is named as the wiretap channel. It is assumed that (i) all users are equipped with a single antenna; (ii) both links are independent and subjected to Fox's $H$-function fading; (iii) a perfect channel state information (CSI) is available at all users.

As a result, the received SNRs at B and E are denoted as $\gamma_k, k \in \{B,E \}$, which follow Fox's $H$-function PDF, and are respectively given by
\begin{subequations}
\begin{equation} \label{PDF_Non}
f_{B} (\gamma_B)= \kappa_B H_{p_0,q_0}^{m_0,n_0} \left[ {  \lambda_B \gamma_B \left| {\begin{array}{*{20}c}
    {(a_i,A_i)_{i=1:p_0}}   \\
   {(b_l,B_l)_{l=1:q_0}}  \\
\end{array}} \right.} \right], \ \gamma_B >0,
\end{equation}
\begin{equation} \label{CDF_Non}
f_{E}(\gamma_E) = \kappa_E H_{p_1,q_1}^{m_1,n_1} \left[ {  \lambda_E \gamma_E \left| {\begin{array}{*{20}c}
    {(c_i,C_i)_{i=1:p_1}}   \\
   {(d_l,D_l)_{l=1:q_1}}  \\
\end{array}} \right.} \right], \ \gamma_E >0.
\end{equation}
\end{subequations}
\subsection{Problem Formulation}
According to \cite{4529264}, the secrecy capacity over fading wiretap channels is defined as the difference between the main channel capacity $C_M = \log_2(1+\gamma_B)$ and the wiretap channel capacity $C_W = \log_2(1+\gamma_E)$ as follows
\begin{equation} \label{Cs_de}
    C_s=
   \begin{cases}
    C_M  - C_W,  &\mbox{$\gamma_B>\gamma_E$}\\
   0, &\mbox{otherwise}.
   \end{cases}
\end{equation}
In other words, a positive secrecy capacity can be assured if and only if the received SNR at Bob has a superior quality than that at Eve's.
\subsubsection{Secrecy Outage Probability}
The outage probability of the secrecy capacity is defined as the probability that the secrecy capacity $C_s$ falls below the target secrecy rate $R_t$, i.e.,
\begin{equation}\label{Definition_Pout}
\mathcal{P}_{out}(R_s)=Pr(C_s<R_t).
\end{equation}
Technically speaking, SOP can be conceptually explained as two cases: (i) $C_s < R_t$ whilst positive secrecy capacity is surely guaranteed; (ii) secrecy outage definitely happens when $C_S$ is non-positive. To this end, (\ref{Definition_Pout}) can be rewritten as follows \cite{7313027,7881216},
\begin{equation} \label{Pout_def}
\begin{split}
\mathcal{P}_{out}(R_s)&  = \mathcal{P}r ( \gamma_B \le R_s \gamma_E + R_s -1 ) \\
&= \int_0^{\infty}  F_{B}(\gamma_0) f_{E}(\gamma_E) d\gamma_E,
\end{split}
\end{equation}
where $R_s = 2^{R_t}$, $\gamma_0 =  R_s \gamma_E + \mathcal{W}$, and $\mathcal{W} = R_s -1$. 

The SOP characterizes the probability of failure to achieve a reliable and secure transmission. In addition, it shows that PLS can be achieved by fading alone, even when Eve has a better average SNR than Bob.
\subsubsection{Probability of Non-Zero Secrecy Capacity}
The PNZ refers to the event that the positive secrecy capacity can be surely achieved, namely $Pr(C_s >0)$, thus respecting its definition, (\ref{Cs_de}) can be further rewritten as follows, 
\begin{equation} \label{def_Pspc}
\begin{split}
\mathcal{P}_{nz}& =Pr(\gamma_B > \gamma_E)= \int_0^{\infty} f_{B}(\gamma_B) F_{E}(\gamma_B) d\gamma_B .
\end{split}
\end{equation}
\subsubsection{Average Secrecy Capacity}
average secrecy capacity provides a mathematical indicator of the capacity limit for a given constraint of perfect secrecy.

By using some simple mathematical manipulations, the ASC can be further re-expressed as the sum of three terms, which are given by \cite{7342908}
\begin{equation} \label{ASC_DEFI}
\begin{split}
\bar{C}_s  = &\underbrace{\int_0^\infty \log_2(1+\gamma_B) f_B(\gamma_B) F_E(\gamma_B) d \gamma_B}_{I_1} \\
& + \underbrace{\int_0^\infty\hspace{-1ex}  \log_2(1+\gamma_E) f_E(\gamma_E) F_B(\gamma_E) d \gamma_E}_{I_2} \\
& - \underbrace{\int_0^\infty \hspace{-1ex} \log_2(1+\gamma_E) f_E(\gamma_E)  d \gamma_E }_{I_3}.
\end{split}
\end{equation}
For the brevity of the following derivations, let $g_k(\gamma_k) = \ln(1+\gamma_k) f_B(\gamma_k) $.
\section{Secrecy Metrics Characterization } \label{Ana_noncoll}
To begin the characterization of the secrecy performance over Fox's $H$-function fading channels, one useful and unified theorem is first provided. This theorem is essentially beneficial to the acquisition of the final closed-form expressions for the aforementioned three secrecy metrics.
\begin{theorem} \label{theorem1}
Consider a general fading channel where the received SNR's PDF is $f(\gamma)$ and another function $u(\gamma)$. Suppose their Mellin transforms are $\mathcal{M}[f(\gamma),s]$ and $\mathcal{M}[u(\gamma),s]$, respectively. If the Mellin transform of $u(\gamma)$ exists, then by using \textit{Parseval's formula} for Mellin transform \cite[eq. (8.3.23)]{debnath2014integral}, we have
\begin{equation}
\hspace{-1.5ex}\int _ 0^\infty f(\gamma) u(\gamma) d\gamma = \frac{1}{2\pi j} \int_{\mathcal{L}} \mathcal{M}[f(\gamma),s] \mathcal{M}[u(\gamma),1-s]ds,
\end{equation}
where ${\mathcal{L}}$ is the integration path from $\upsilon-j\infty$ to $\upsilon + j\infty$, and $\upsilon$ is a constant.
\end{theorem}
The aforementioned Theorem is recalled to make a basis for the following derivations. To this end, we have the following remark.
\begin{remark} 
The SOP, PNZ, and ASC over Fox's $H$-function fading wiretap channels are respectively given by
\begin{subequations}
\begin{equation} \label{Pout_theorem1}
\mathcal{P}_{out} = \frac{1}{2\pi j} \int_{\mathcal{L}_1} \mathcal{M}[F_B(\gamma_0),1-s]\mathcal{M}[f_E(\gamma_E),s]ds,
\end{equation}
\begin{equation} \label{Pnz_theorem1}
\mathcal{P}_{nz} =\frac{1}{2\pi j} \int_{\mathcal{L}_1}\mathcal{M}[F_E(\gamma_B),1-s]\mathcal{M}[f_B(\gamma_B),s]ds,
\end{equation}
\begin{equation} \label{Cs_theorem1}
\begin{split}
\bar{C}_s & = \frac{1}{2\pi j} \int_{\mathcal{L}_1} \mathcal{M}[g_B(\gamma_E),1-s]\mathcal{M}[F_E(\gamma_B),s]ds  \\
&  + \frac{1}{2\pi j} \int_{\mathcal{L}_1} \mathcal{M}[g_E (\gamma_E),1-s] \mathcal{M}[F_B(\gamma_E),s] ds \\
& -\frac{1}{2\pi j} \int_{\mathcal{L}_1} \mathcal{M}[f_E(\gamma_E),1-s]\mathcal{M}[\ln(1 + \gamma_E),s] ds
\end{split}
\end{equation}
\end{subequations}
\end{remark}
\begin{proof}
Recalling (\ref{Pout_def}), (\ref{def_Pspc}), and (\ref{ASC_DEFI}), and then using \textbf{Theorem} \ref{theorem1}, the proofs for (\ref{Pout_theorem1}), (\ref{Pnz_theorem1}), and (\ref{Cs_theorem1}) are directly accomplished. 
\end{proof}
\subsection{SOP Characterization}
\subsubsection{Exact SOP Characterization}
\begin{theorem} \label{theorem_pout}
The SOP over Fox's $H$-function fading wiretap channels is given by (\ref{pout_theorem}), shown at the top of this page. 
\begin{figure*}[!t]
\setcounter{MYtempeqncnt}{\value{equation}}
\setcounter{equation}{12}
\begin{equation} \label{pout_theorem}
\begin{split}
\hspace{-1ex}\mathcal{P}_{out} &=1-  \frac{\kappa_B \kappa_E \mathcal{W}}{\lambda_B R_s}   H_{1,0:q_1,p_1+1:q_0,p_0+1}^{0,1:n_1+1,m_1:n_0,m_0} \left[ {  \frac{R_s }{ \lambda_E \mathcal{W}},\frac{1}{\lambda_B\mathcal{W}} \left|\hspace{-1ex} {\begin{array}{*{20}c}
    {(2,1,1)}   \\
   {-}  \\
\end{array}} \hspace{-1ex}\right. \hspace{-0.5ex}
\left|\hspace{-1ex} {\begin{array}{*{20}c}
   {(1-d_l,D_l)_{l=1:q_1}}  \\
   {(1,1),(1-c_i,C_i)_{i=1:p_1}}   \\
\end{array}} \hspace{-1ex}\right. 
\left|\hspace{-1ex} {\begin{array}{*{20}c}
    {(1-b_l-B_l,B_l)_{l=1:q_0}}  \\
    {(1 - a_i-A_i,A_i)_{i=1:p_0},(0,1)}   \\
\end{array}} \right.} \hspace{-1ex} \right] ,
\end{split}
\end{equation}
\hrulefill
\end{figure*}
\end{theorem}
\begin{proof}
See Appendix \ref{appendix1}.
\end{proof}
\subsubsection{Lower Bound of SOP}
As $\bar{\gamma}_B$ and $\bar{\gamma}_E$ tend to $\infty$, we have 
\begin{equation}
\begin{split}
\mathcal{P}_{out} & = Pr\left(\log_2 \left(\frac{1 + \gamma_B}{1 + \gamma_E} \right) < R_t \right)\\
& \approx Pr \underbrace{\left(\log_2 \left(\frac{ \gamma_B}{ \gamma_E} \right) <R_t \right)}_{ \mathcal{P}_{out}^{L}} 
= \int_0^\infty F_{B}(R_s y) f_E(y) dy  .
\end{split}
\end{equation}
\begin{proposition}
As $\bar{\gamma}_B$ and $\bar{\gamma}_E$ tend to $\infty$, the lower bound of the SOP over Fox's $H$-function fading channels is given by (\ref{PoutAsy_theorem}), shown at the top of next page.
\begin{figure*}[!t]
\setcounter{MYtempeqncnt}{\value{equation}}
\setcounter{equation}{14}
\begin{equation} \label{PoutAsy_theorem}
\mathcal{P}_{out}^{L} = 1 - \frac{\kappa_B\kappa_E}{\lambda_B\lambda_E} H_{p_1+q_2+1,q_1+p_2+1}^{m_1+n_2+1,n_1+m_2 } \left[ { \frac{\lambda_B R_s}{\lambda_E} \left| {\begin{array}{*{20}c}
    {(a_i+A_i,A_i)_{i=1:n_1},(1-d_l-D_l,D_l)_{l=1:q_2},(a_i+A_i,A_i)_{i=n_1+:p_1},(1,1)}   \\
   {(0,1),(b_l+B_l,B_l)_{l=1:m_1},(1-c_i-C_i,C_i)_{i=1:p_2}, (b_l+B_l,B_l)_{l=m_1+1:q_1}}  \\
\end{array}} \right.} \right].
\end{equation} 
\hrulefill
\vspace{-0.45cm}
\end{figure*}
\end{proposition}
\begin{proof}
By applying the Mellin transform of the product of two Fox's $H$-function \cite[eq. (2.25.1.1)]{prudnikov1990integrals}, the proof is accomplished.
\end{proof}
\subsection{PNZ Characterization}
\begin{theorem}
The PNZ over Fox's $H$-function wiretap fading channels is given by (\ref{Pnz_theorem}), shown at the top of next page.
\begin{figure*}[!t]
\setcounter{MYtempeqncnt}{\value{equation}}
\setcounter{equation}{15}
\begin{equation} \label{Pnz_theorem}
\begin{split}
\mathcal{P}_{nz}   = \frac{\kappa_B \kappa_E}{\lambda_B\lambda_E} H_{p_0+q_1+1,q_0+p_1+1}^{m_1+n_0,n_1+m_0+1} \left[ {  \frac{\lambda_E}{\lambda_B} \left| {\begin{array}{*{20}c}
    {(1,1),(c_i+C_i,C_i)_{i=1:p_1},(1-b_l-B_l,B_l)_{l=1:q_0},(c_i+C_i,C_i)_{i=n_1+1:p_1}}   \\
   {(d_l+D_l,D_l)_{l = 1:m_1},(1-a_i-A_i,A_i)_{i=1:p_0},(0,1),(d_l+D_l,D_l)_{l = m_1 +1:q_1} }  \\
\end{array}} \right.} \right],
\end{split}
\end{equation}
\hrulefill
\vspace{-0.4cm}
\end{figure*}
\end{theorem}
\begin{proof}
According to (\ref{Pnz_theorem1}), $\mathcal{M}[F_E(\gamma_B),1-s]$ and $\mathcal{M}[f_B(\gamma_B),s] $ are separately given by
\begin{subequations}
\begin{equation} \label{Me_pnz1}
\begin{split}
\mathcal{M}[F_E(\gamma_B),1-s] = \frac{\kappa_E}{\lambda_E^{2-s}}\Theta_E^{F}(1-s),
\end{split}
\end{equation}
\begin{equation} \label{Me_pnz2}
\mathcal{M}[f_B(\gamma_B),s] = \frac{\kappa_B}{\lambda_B^{s}}\Theta_B^f(s).
\end{equation}
\end{subequations}

Next, substituting (\ref{Me_pnz1}) and (\ref{Me_pnz2}) into (\ref{Pnz_theorem1}), yields the following result
\begin{equation}
\begin{split}
&\mathcal{P}_{nz} = \frac{\kappa_B \kappa_E}{2 \lambda_E^2\pi j} \int_{\mathcal{L}_1} \Theta_B^f(s) \Theta_E^{F}(1-s) \left( \frac{\lambda_B}{\lambda_E} \right)^{-s}ds ,
\end{split}
\end{equation}
Subsequently, directly applying the definition of univariate Fox's $H$-function, the proof is achieved.

Alternatively, we provide another method to prove (\ref{Pnz_theorem}). Revisiting (\ref{def_Pspc}) and directly replacing $f_B(\gamma_B)$ and $F_E(\gamma_B)$ with their expressions, we have
\begin{equation}
\begin{split}
& \mathcal{P}_{nz} = \frac{\kappa_B \kappa_E}{\lambda_E} \int_0^\infty  H_{p_0,q_0}^{m_0,n_0} \left[ {  \lambda_B \gamma_B \left| {\begin{array}{*{20}c}
    {(a_i,A_i)_{i=1:p_0}}   \\
   {(b_l,B_l)_{l=1:q_0}}  \\
\end{array}} \right.} \right] \\
& \hspace{1ex} \times  H_{p_1+1,q_1+1}^{m_1,n_1+1} \left[ {  \lambda_E \gamma_B \left| {\begin{array}{*{20}c}
    {(1,1),(c_i+C_i,C_i)_{i=1:p_1}}   \\
   {(d_l+D_l,D_l)_{l = 1:q_1},(0,1)}  \\
\end{array}} \right.} \right] d\gamma_B,
\end{split}
\end{equation}
where the last step is derived by using the Mellin transform of the product of two Fox's $H$-function \cite[eq. (2.25.1.1)]{prudnikov1990integrals}.
\end{proof}
\subsection{ASC Characterization}
\begin{theorem} \label{theorem_asc}
The ASC over Fox's $H$-function wiretap fading channels is given by 
\begin{equation}
\bar{C}_s =\frac{1}{\ln (2)}(I_1 + I_2 - I_3),
\end{equation}
where $I_1$ and $I_2$ are respectively given by (\ref{I1_final}) and (\ref{I2_final}), shown at the top of this page, and
\begin{figure*}[t] 
\setcounter{MYtempeqncnt}{\value{equation}}
\setcounter{equation}{20}
\begin{subequations}
\begin{equation} \label{I1_final}
I_1 = \frac{\kappa_B \kappa_E}{ \lambda_B \lambda_E } H_{q_0,p_0:2,2:p_1+1,q_1 +1}^{n_0,m_0:1,2:m_1,n_1+1} \left[ {  \frac{1}{\lambda_B}, \frac{\lambda_E}{\lambda_B} \left| {\begin{array}{*{20}c}
    {(1-b_l -B_l;B_l,B_l)_{l=1:q_0}}   \\
   {(1-a_i -A_i;A_i,A_i)_{i=1:p_0}}  \\
\end{array}} \right. 
\left| {\begin{array}{*{20}c}
    {(1,1),(1,1)}   \\
   {(1,1),(0,1)}  \\
\end{array}} \right. 
\left| {\begin{array}{*{20}c}
    {(1,1),(c_i+C_i,C_i)_{i=1:q_1}}  \\
    {(d_l+D_l,D_l)_{l=1:p_1},(0,1)}   \\
\end{array}} \right.} \right],
\end{equation} 
\begin{equation} \label{I2_final}
I_2= \frac{\kappa_B \kappa_E}{ \lambda_B \lambda_E } H_{q_1,p_1:2,2:p_0+1,q_0 +1}^{n_1,m_1:1,2:m_0,n_0+1} \left[ {  \frac{1}{\lambda_E}, \frac{\lambda_B}{\lambda_E} \left| {\begin{array}{*{20}c}
    {(1-d_l -D_l;D_l,D_l)_{l=1:q_1}}   \\
   {(1-c_i -C_i;C_i,C_i)_{i=1:p_1}}  \\
\end{array}} \right. 
\left| {\begin{array}{*{20}c}
    {(1,1),(1,1)}   \\
   {(1,1),(0,1)}  \\
\end{array}} \right. 
\left| {\begin{array}{*{20}c}
    {(1,1),(a_i+A_i,A_i)_{i=1:p_0}}   \\
   {(b_l + B_l,B_l)_{l=1:q_0},(0,1)}  \\
\end{array}} \right.} \right],
\end{equation}
\end{subequations}
\hrulefill
\end{figure*}
\begin{equation} \label{I3_final}
I_3 =  \frac{\kappa_E}{\lambda_E} H_{q_1+2,p_1+2}^{n_1 + 1,m_1 + 2} \left[ { \frac{1}{ \lambda_E} \left|\hspace{-0.6ex} {\begin{array}{*{20}c}
    {(1,1),(1,1),(1-d_l-D_l,D_l)_{l=1:p_1}}   \\
   {(1,1),(1-c_i-C_i,C_i)_{i=1:q_1},(0,1)}  \\
\end{array}} \right.} \hspace{-0.6ex} \right].
\end{equation}
\end{theorem}
\begin{proof}
See Appendix \ref{Appen_asc}
\end{proof}
\subsection{Special Cases}
Accommodating the closed-form expressions for secrecy performance metrics in the corresponding entries in Table \ref{T1}, directly yields the results, as displayed in Table \ref{T2}. After some simple algebraic manipulations, one can observe the obtained results herein are consistent with the existing works \cite{7094262,7856980,OutageLongLett,LongFisherF}.
\begin{table*}[!t]    
\renewcommand{\arraystretch}{1.15}
\caption{Exact expressions of $\mathcal{P}_{out}$, $\mathcal{P}_{nz}$ and $\bar{C}_s$ for different special cases of Fox's $H$-function distribution} 
\label{T2} 
\hspace{0.5cm}  
  \centering  
  \begin{tabular} {p{22pt} p{470pt}}\toprule \hline
& $\mathcal{P}_{out} =   1- \frac{\kappa_B \kappa_E \mathcal{W}}{\lambda_B R_s}  H_{1,0:1,1:1,1}^{0,1:1,1:0,1} \left[ {  \frac{R_s }{ \lambda_E \mathcal{W}},\frac{1}{\lambda_B  \mathcal{W}} \left| {\begin{array}{*{20}c}
    {(2,1,1)}   \\
   {-}  \\
\end{array}} \right. 
\left| {\begin{array}{*{20}c}
   {(1 - \mu_E + \frac{1}{\alpha_E},\frac{1}{\alpha_E})}  \\
   {(1,1)}   \\
\end{array}} \right. 
\left| {\begin{array}{*{20}c}
   {(1 - \mu_B,\frac{1}{\alpha_B})}  \\
   {(0,1)}   \\
\end{array}} \right.} \right]$ \\ \specialrule{0em}{1pt}{1pt} \cdashline{2-2} \specialrule{0em}{1pt}{1pt}
   & $\mathcal{P}_{nz} =  \frac{\kappa_B \kappa_E}{\lambda_B\lambda_E} H_{2,2}^{2,1} \left[ {  \frac{\lambda_E}{\lambda_B} \left| {\begin{array}{*{20}c}
    {(1-\mu_B , \frac{1}{\alpha_B}),(1,1)}   \\
   {(\mu_E,\frac{1}{\alpha_E}),(0,1)} \\
   \end{array}} \right.} \right]$ \\ \specialrule{0em}{1pt}{1pt} \cdashline{2-2} \specialrule{0em}{1pt}{1pt}
\textbf{$\alpha-\mu$ }  & $\bar{C}_s =  \frac{\kappa_B \kappa_E}{\lambda_B \lambda_E} H_{1,0:2,2:1,2}^{0,1:1,2:1,1} \left[ {  \frac{1}{\lambda_B}, \frac{\lambda_E}{\lambda_B} \left| {\begin{array}{*{20}c}
    {(1-\mu_B ;\frac{1}{\alpha_B},\frac{1}{\alpha_B})}   \\
   {-}  \\
\end{array}} \right. 
\left| {\begin{array}{*{20}c}
    {(1,1),(1,1)}   \\
   {(1,1),(0,1)}  \\
\end{array}} \right. 
\left| {\begin{array}{*{20}c}
    {(1,1)}   \\
   {(\mu_E,\frac{1}{\alpha_E}),(0,1)}  \\
\end{array}} \right.} \right]  $ \\ \specialrule{0em}{1pt}{1pt}
 & $\hspace{5ex} + \frac{\kappa_B \kappa_E}{\lambda_B\lambda_E} H_{1,0:2,2:1,2}^{0,1:1,2:1,1} \left[ {  \frac{1}{\lambda_E}, \frac{\lambda_B}{\lambda_E} \left| {\begin{array}{*{20}c}
    {(1-\mu_E ;\frac{1}{\alpha_E},\frac{1}{\alpha_E})}   \\
   {-}  \\
\end{array}} \right. 
\left| {\begin{array}{*{20}c}
    {(1,1),(1,1)}   \\
   {(1,1),(0,1)}  \\
\end{array}} \right. 
\left| {\begin{array}{*{20}c}
    {(1,1)}   \\
   {(\mu_B,\frac{1}{\alpha_B}),(0,1)}  \\
\end{array}} \right.} \right] $ \\ \specialrule{0em}{1pt}{1pt}
&$\hspace{5ex} - \frac{\kappa_E}{\lambda_E} H_{3,2}^{1,3} \left[ {  \frac{1}{\lambda_E} \left| {\begin{array}{*{20}c}
    {(1,1),(1,1),(1-\mu_E,\frac{1}{\alpha_E})}   \\
   {(1,1),(0,1)}  \\
\end{array}} \right.} \right]$ \\  \midrule
& $\mathcal{P}_{out} = 1- \frac{\kappa_B \kappa_E \mathcal{W}}{\lambda_B R_s}  H_{1,0:1,2:1,2}^{0,1:2,1:1,1} \left[ {  \frac{R_s }{ \lambda_E \mathcal{W}},\frac{1}{\lambda_B  \mathcal{W}} \left| {\begin{array}{*{20}c}
    {(2,1,1)}   \\
   {-}  \\
\end{array}} \right. 
\left| {\begin{array}{*{20}c}
   {(2 - m_E,1)}  \\
   {(1,1),(1 + m_{E,s},1)}   \\
   \end{array}} \right. 
\left| {\begin{array}{*{20}c}
   {(1 - m_B,1 )}  \\
   {(m_{B,s},1),(0,1)}   \\
\end{array}} \right.} \right]$ \\ \specialrule{0em}{1pt}{1pt} \cdashline{2-2} \specialrule{0em}{1pt}{1pt}
   & $\mathcal{P}_{nz} =  \frac{\kappa_B \kappa_E}{\lambda_B\lambda_E} H_{3,3}^{2,3} \left[ {  \frac{\lambda_E}{\lambda_B} \left| {\begin{array}{*{20}c}
    {(1,1),(-m_{B,s},1),(1-m_E,1),(0,1)}   \\
   {(m_{E},1),(-1,1),(m_{E,s},1),(0,1)} \\
   \end{array}} \right.} \right]$ \\ \specialrule{0em}{1pt}{1pt} \cdashline{2-2} \specialrule{0em}{1pt}{1pt}
   \specialrule{0em}{1pt}{1pt}
\textbf{F-S $\mathcal{F}$}  & $\bar{C}_s =   \frac{\kappa_B \kappa_E}{\lambda_B \lambda_E} H_{1,1:2,2:2,2}^{1,1:1,2:1,2} \left[ {  \frac{1}{\lambda_B}, \frac{\lambda_E}{\lambda_B} \left| {\begin{array}{*{20}c}
    {(m_B;1,1)}   \\
   {(m_{B,s};1,1)}  \\
\end{array}} \right. 
\left| {\begin{array}{*{20}c}
    {(1,1),(1,1)}   \\
   {(1,1),(0,1)}  \\
\end{array}} \right. 
\left| {\begin{array}{*{20}c}
    {(1,1),(1-m_{E,s},1)}   \\
   {(m_E,1),(0,1)}  \\
\end{array}} \right.} \right]  $ \\ \specialrule{0em}{1pt}{1pt}
 & $\hspace{5ex} + \frac{\kappa_B \kappa_E}{\lambda_B\lambda_E} H_{1,1:2,2:2,2}^{1,1:1,2:1,2} \left[ {  \frac{1}{\lambda_E}, \frac{\lambda_B}{\lambda_E} \left| {\begin{array}{*{20}c}
    {(m_E ;1,1)}   \\
   {(m_{E,s};1,1)}  \\
\end{array}} \right. 
\left| {\begin{array}{*{20}c}
    {(1,1),(1,1)}   \\
   {(1,1),(0,1)}  \\
\end{array}} \right. 
\left| {\begin{array}{*{20}c}
    {(1,1),(1-m_{E,s},1)}   \\
   {(m_E,1),(0,1)}  \\
\end{array}} \right.} \right] $ \\  \specialrule{0em}{1pt}{1pt}
&$\hspace{5ex} - \frac{\kappa_E}{\lambda_E} H_{3,3}^{2,3} \left[ {  \frac{1}{\lambda_E} \left| {\begin{array}{*{20}c}
   {(1,1),(1,1),(1-m_E,1)}  \\
   {(1,1),(m_{E,s},1),(0,1)}   \\
\end{array}} \right.} \right]$ \\ \midrule
& $\mathcal{P}_{out} =   1-\frac{\kappa_B \kappa_E \mathcal{W}}{\lambda_B R_s}   H_{1,0:2,1:2,1}^{0,1:1,1:0,2} \left[ {  \frac{R_s }{ \lambda_E \mathcal{W}},\frac{1}{\mathcal{W}\lambda_B} \left| {\begin{array}{*{20}c}
    {(2,1,1)}   \\
   {-}  \\
\end{array}} \right. 
\left| {\begin{array}{*{20}c}
   {(1 - m_E + \frac{1}{\xi_{E}},\frac{1}{\xi_{E}}),(1 - m_{sE} + \frac{1}{\xi_{sE}},\frac{1}{\xi_{sE}})}  \\
   {(1,1)}   \\
\end{array}} \right. 
\left| {\begin{array}{*{20}c}
   {(1 - m_B ,\frac{1}{\xi_{B}}),(1 - m_{sB},\frac{1}{\xi_{sB}})}  \\
   {(0,1)}   \\
\end{array}} \right.} \right]$\\ \specialrule{0em}{1pt}{1pt} \cdashline{2-2} \specialrule{0em}{1pt}{1pt}
 &$\mathcal{P}_{nz} = \frac{\kappa_B \kappa_E}{\lambda_B\lambda_E} H_{3,3}^{2,3} \left[ {  \frac{\lambda_E}{\lambda_B} \left| {\begin{array}{*{20}c}
    {(1,1),(1-m_B ,\frac{1}{\xi_B} ),(1-m_{sB},\frac{1}{\xi_{sB}} )}   \\
   {(m_E ,\frac{1}{\xi_E}),(m_{sE} ,\frac{1}{\xi_{sE}}),(0,1)}  \\
\end{array}} \right.} \right] $   \\ \specialrule{0em}{1pt}{1pt} \cdashline{2-2} \specialrule{0em}{1pt}{1pt}
\textbf{EGK} & $\bar{C}_s = \frac{\kappa_B \kappa_E}{\lambda_B \lambda_E} H_{2,0:2,2:1,3}^{0,2:1,2:3,0} \left[ {  \frac{1}{\lambda_B}, \frac{\lambda_E}{\lambda_B}  \left| {\begin{array}{*{20}c}
    {(1-m_B;\frac{1}{\xi_B},\frac{1}{\xi_B}),(1-m_{sB},\frac{1}{\xi_{sB}} ,\frac{1}{\xi_{sB}})}   \\
   {-}  \\
\end{array}} \right. 
\left| {\begin{array}{*{20}c}
    {(1,1),(1,1)}   \\
   {(1,1),(0,1)}  \\
\end{array}} \right. 
\left| {\begin{array}{*{20}c}
    {(1,1)}   \\
   {(m_E ,\frac{1}{\xi_E}),(m_{sE},\frac{1}{\xi_{sE}}),(0,1)}  \\
\end{array}} \right.} \right]  $ \\ \specialrule{0em}{1pt}{1pt}
&$ \hspace{3ex} + \frac{\kappa_B \kappa_E}{\lambda_B \lambda_E} H_{2,0:2,2:1,3}^{0,2:1,2:3,0} \left[ {  \frac{1}{\lambda_E}, \frac{\lambda_B}{\lambda_E}  \left| {\begin{array}{*{20}c}
    {(1-m_E;\frac{1}{\xi_E},\frac{1}{\xi_E}),(1-m_{sE};\frac{1}{\xi_{sE}},\frac{1}{\xi_{sE}})}   \\
   {-}  \\
\end{array}} \right. 
\left| {\begin{array}{*{20}c}
    {(1,1),(1,1)}   \\
   {(1,1),(0,1)}  \\
\end{array}} \right. 
\left| {\begin{array}{*{20}c}
    {(1,1)}   \\
   {(m_B ,\frac{1}{\xi_B}),(m_{sB},\frac{1}{\xi_{sB}}),(0,1)}  \\
\end{array}} \right.} \right]$ \\
&$\hspace{3ex}  -\frac{\kappa_E}{\lambda_E} H_{2,4}^{4, 1} \left[ {  \frac{1}{\lambda_E}\left| {\begin{array}{*{20}c}
    {(1,1),(1,1)}   \\
   {(1,1),(m_E - \frac{1}{\xi_E},\frac{1}{\xi_E}),(m_{sE}-\frac{1}{\xi_{sE}},\frac{1}{\xi_{sE}}),(0,1)}  \\
\end{array}} \right.} \right]$ \\ \specialrule{0em}{1pt}{1pt}
\midrule
 \bottomrule
  \end{tabular}  
  \vspace{-0.3cm}
\end{table*}
\section{Asymptotic Secrecy Metrics Characterization } \label{Ana_Asym}
The obtained secrecy expressions are given in terms of either univariate or bivariate Fox's $H$-function. In order to provide more insights at high or low SNR regime, the asymptotic behavior of the three aforementioned secrecy metrics are developed in this section.

According to \cite{7370883}, expansions of the univariate and bivariate Fox’s $H$-functions can be derived by evaluating the residue of the corresponding integrands at the closest poles to the
contour, namely, the minimum pole on the right for large Fox’s $H$-function arguments and the maximum pole on the left for small ones.
\subsection{Asymptotic SOP}
The lower bound of the SOP is still expressed in terms of Fox's $H$-function, in order to study the asymptotic behavior of the SOP, the lower bound is further simplified by expanding the univariate Fox's $H$-function. Consequently, at high $\bar{\gamma}_B$ regime, we have $\frac{1}{\lambda_B} \rightarrow \infty$. By using the expanding rule, the asymptotic SOP is given by (\ref{SOP_ASYMP}), shown at the top of next page.
\begin{figure*}[!t]
\setcounter{MYtempeqncnt}{\value{equation}}
\setcounter{equation}{22}
\begin{equation} \label{SOP_ASYMP}
\begin{split}
\mathcal{P}_{out}^L \approx& 1-  \frac{\kappa_B \kappa_E }{\lambda_B \lambda_E}\frac{ \Gamma(\tau) \prod \limits_{l=1,l\neq g}^{m_0}\Gamma(b_l+B_l + B_l \tau)\prod \limits_{i=1}^{n_1}\Gamma(1 - c_i- C_i + C_i \tau) \prod \limits_{i=1}^{n_0}\Gamma(1 - a_i- A_i + A_i \tau)}{\Gamma(1+\tau) \prod \limits_{l=m_0+1}^{q_0}\Gamma(1-b_l-B_l - B_l \tau)\prod \limits_{i=n_0+1}^{p_0}\Gamma( a_i+ A_i + A_i \tau)\prod \limits_{i=n_1+1}^{p_2}\Gamma( c_i+ C_i - C_i \tau)} \\
& \times \frac{\prod \limits_{l=1}^{m_1}\Gamma(d_l+D_l - D_l \tau)}{\prod \limits_{l=m_1+1}^{q_1}\Gamma(1-d_l-D_l + D_l \tau)}  \left( \frac{\lambda_E}{\lambda_B R_s}\right)^\tau , \ \rm{where} \  \tau = \max \limits_{l=1:m_0} \left(- \frac{b_l+B_l}{B_l} \right), g = \argmax\limits_{l=1:m_0}\left(- \frac{b_l+B_l}{B_l} \right).
\end{split}
\end{equation}
\hrulefill
\vspace{-0.4cm}
\end{figure*}

Taking the case of $\alpha-\mu$ distribution as an example, the lower bound of the SOP is given by
\begin{equation}
\mathcal{P}_{out}^L = 1 - \frac{\kappa_B \kappa_E}{\lambda_B\lambda_E } H_{2,2}^{2,1} \left[ { \frac{ \lambda_B R_s} {\lambda_E}\left| {\begin{array}{*{20}c}
    {\left(1-\mu_E,\frac{1}{\alpha_E}\right),(1,1)}   \\
   {(0,1),\left(\mu_B,\frac{1}{\alpha_B}\right)}  \\
\end{array}} \right.} \right].
\end{equation}
For the sake of high accuracy, the asymptotic SOP at high $\bar{\gamma}_B$ regime is evaluated at $\tau = 0 $ and $\tau = -\alpha_B\mu_B$, and is given by \cite{OutageLongLett}
\begin{equation}
\mathcal{P}_{out} \approx \frac{\Gamma\left(\frac{\alpha_B\mu_B}{\alpha_E} +\mu_E \right)}{\Gamma(1+\mu_B)\Gamma(\mu_E)}\left( \frac{R_s\lambda_B}{\lambda_E}\right)^{\alpha_B\mu_B}.
\end{equation} 
\subsection{Asymptotic PNZ}
The asymptotic PNZ at high or low $\bar{\gamma}_B$ regime, is computed by evaluating the residues of analytical PNZ, given in (\ref{Pnz_theorem}). According to \cite{8338118}, Fox's $H$-function can be further simplified by choosing the dominate term of the Mellin-Barnes type integral. As such, we can evaluate the residue of the PNZ at low $\bar{\gamma}_B$ regime, at the point  
\begin{subequations}
\begin{equation}
\tau =  \min\limits_{l=1:m_1,i=1:n_0} \left(-\frac{d_l + D_l}{D_l},\frac{a_i +A_i-1}{A_i} \right),
\end{equation}
\begin{equation}
g = \argmin\limits_{l=1:m_1,i=1:n_0}\left(-\frac{d_l + D_l}{D_l},\frac{a_i +A_i-1}{A_i}\right).
\end{equation}
\end{subequations}
Assuming the case of a simple pole, the asymptotic PNZ is thereafter given in (\ref{PNZ_ASYMP}).
\begin{figure*}[!t] 
\setcounter{MYtempeqncnt}{\value{equation}}
\setcounter{equation}{26}
\begin{equation} \label{PNZ_ASYMP}
\mathcal{P}_{nz} \approx  \frac{\left( \frac{\lambda_B}{\lambda_E}\right)^s  \frac{\kappa_B \kappa_E}{\lambda_B\lambda_E}\Gamma(-\tau)\prod \limits_{l=1,l\neq g}^{m_1}\Gamma(d_l+D_l + D_l \tau)\prod \limits_{i=1, i\neq g}^{n_0}\Gamma(1 - a_i- A_i + A_i \tau)\prod \limits_{i=1}^{n_1}\Gamma(1 - c_i- C_i - C_i \tau)\prod \limits_{l=1}^{m_0}\Gamma(b_l+B_l - B_l \tau)}{\Gamma(1-\tau)\prod \limits_{i=n_0 + 1}^{p_0}\Gamma(a_i+A_i - A_i \tau)\prod \limits_{l=m_1 +1}^{q_1}\Gamma(1-d_l-D_l - D_l \tau)\prod \limits_{i=n_1 + 1}^{p_1}\Gamma(c_i+C_i + C_i \tau)\prod \limits_{l=m_0 + 1}^{q_0}\Gamma(1-b_l-B_l + B_l \tau)} ,
\end{equation}
\hrulefill
\vspace{-0.2cm}
\end{figure*}

Considering the case of $\alpha-\mu$ as an example, applying the obtained result, the asymptotic PNZ at low $\bar{\gamma}_B$ regime is evaluated at $s = -\alpha_E\mu_E$ and thereafter given by
\begin{equation}
\mathcal{P}_{nz} \approx \frac{\kappa_B \kappa_E}{\lambda_B\lambda_E \mu_E} \Gamma\left( \frac{\alpha_E \mu_E}{\alpha_B} + \mu_B\right)\left(  \frac{\lambda_E}{\lambda_B} \right)^{\alpha_E\mu_E}.
\end{equation}
\subsection{Asymptotic ASC}
By applying the expansion rule, in the case of high $\bar{\gamma}_B$, the asymptotic ASC is given by (\ref{ASC_ASYMP}), which is obtained by individually expanding $I_1$ and $I_2$, respectively.
\begin{figure*}[!t] 
\setcounter{MYtempeqncnt}{\value{equation}}
\setcounter{equation}{28}
\begin{subequations} \label{ASC_ASYMP}
\begin{equation}
\begin{split}
I_1 &\approx \frac{\kappa_B \kappa_E}{\lambda_B \lambda_E} \left[ \ln\left( \frac{1}{\lambda_B} \right) + \sum_{l=1}^{m_0} B_l \Psi_0(b_l + B_l + B_l u) -  \sum_{l=m_l +1}^{q_0} B_l \Psi_0(b_l + B_l + B_l u) -  \sum_{i=1}^{p_0} A_i \Psi_0(a_i + A_i + A_i u)   \right] \\
& \quad \times  \frac{  \left(\frac{\lambda_E}{\lambda_B} \right)^s\Gamma(u)\prod \limits_{l=1,l\neq g}^{m_0} \Gamma( b_l +B_l + B_l u)\prod \limits_{i=1,i\neq g}^{n_1} \Gamma(1- c_i -C_i + C_i u )\prod \limits_{l=1}^{m_1} \Gamma( d_l +D_l - D_l u) }{\Gamma(1+u)\prod \limits_{l=m_1+1}^{q_0} \Gamma(1 - b_l -B_l - B_l u)\prod \limits_{i=n_1+1}^{p_0} \Gamma( a_i +A_i + A_i u)\prod \limits_{i=n_2 +1}^{p_1} \Gamma( c_i +C_i - C_i u)\prod \limits_{l=m_2+1}^{q_1} \Gamma(1- d_l - D_l - D_l u)} ,\\
& {\rm{where}} \ u= \max\limits_{l=1:m_0,i=1:n_1} \left[0, \left(-\frac{b_l + B_l}{B_l} \right)_{l = 1:m_0},  \left(\frac{c_i + C_i - 1}{c_i} \right)_{i = 1:n_1}\right], g = \argmax\limits_{l=1:m_0,i=1:n_1} \left[0, \left(-\frac{b_l + B_l}{B_l} \right)_{l = 1:m_0},  \left(\frac{c_i + C_i - 1}{c_i} \right)_{i = 1:n_1}\right],
\end{split}
\end{equation}
\hrulefill
\begin{equation}
\begin{split}
I_2&  \approx  \prod \limits_{l=1,l \neq g}^{m_0}\Gamma(b_l +B_l - B_l u) \left( \frac{\lambda_E}{\lambda_B}\right)^u  \frac{\kappa_B \kappa_E}{\lambda_B \lambda_E}\frac{\Gamma(u)\prod \limits_{i=1}^{n_0}\Gamma(1-a_i -A_i + A_i u)}{\Gamma(1+u)\prod \limits_{i=1}^{n_0+1}\Gamma(a_i+A_i-A_i u)\prod \limits_{l=1}^{m_0+1}\Gamma(1-b_l-B_l+B_l u)}  \\
&\quad \times H_{q_1+2,p_1+2}^{n_1 + 1,m_1 + 2} \left[ {  \lambda_E \left| {\begin{array}{*{20}c}
    {(1,1),(1,1),(1-d_l-D_l,D_l)_{l=1:p_1},(1,1)}   \\
   {(0,1),(0,1),(d_l +D_l + D_l s),(1-c_i-C_i,C_i)_{i=1:q_1}}  \\
\end{array}} \right.} \right],\ {\rm{where}} \ u = 1 + \min\left(\frac{b_l }{B_l}\right)_{l=1:m_0}, g = \argmin\limits_{l=1:m_0}\left(\frac{b_l }{B_l}\right).
\end{split}
\end{equation}
\end{subequations}
\hrulefill
\vspace{-0.2cm}
\end{figure*}
The detailed proof for (\ref{ASC_ASYMP}) is referenced to Appendix. \ref{Appendix_ASC}.

Similarly, taking the case of $\alpha-\mu$ as an example, we get the asymptotic ASC at high $\bar{\gamma}_B$ regime as
\begin{subequations}
\begin{equation}
I_1 \approx \frac{\kappa_B \kappa_E}{\lambda_B \lambda_E}\Gamma(\mu_B) \Gamma(\mu_E) \left[\frac{\Psi_0(\mu_B)}{\alpha_B} - \ln(\lambda_B) \right],
\end{equation}
\begin{equation}
I_2 \approx \frac{\kappa_B \kappa_E \left(\frac{\lambda_B}{\lambda_E} \right)^{\alpha_B \mu_B}}{\mu_B\lambda_B \lambda_E}   H_{2,3}^{3,1} \left[ {  \lambda_E \left| \hspace{-1ex}{\begin{array}{*{20}c}
    {(0,1),(1,1)}   \\
   {(\mu_E+\frac{\alpha_B\mu_B}{\alpha_E},\frac{1}{\alpha_E}),(0,1),(0,1)}  \\
\end{array}} \right.}\hspace{-1ex} \right].
\end{equation}
\end{subequations}

\vspace{-0.5cm}
\section{Colluding Eavesdropping Scenario} \label{Label_Colluding}
In this section, we mainly focus on the secrecy issue when multiple eavesdroppers appear and work in a cooperative manner. 
\subsection{System Model}
Consider the scenario that $L$ eavesdroppers are in the presence and work cooperatively to wiretap the main link. It is assumed that all $L$ eavesdroppers are single-antenna equipped, and undergoes independent fading conditions. As a result of collusion \cite{8327926}, the so-called eavesdropper is assumed to either use the MRC or the SC scheme. All the wiretap links and main link undergo independent Fox's $H$-function fading channels. Consequently, the instantaneous received SNR at the so-called $L$-colluding eavesdropper with MRC scheme is given by 
\begin{equation} \label{gamma_Co_MRC}
\gamma_C= \sum \limits_{r =1}^L \gamma_{e,r},
\end{equation}
or with SC scheme
\begin{equation} \label{gamma_Co_SC}
\gamma_C= \max \{ \gamma_{e,1},\cdots,\gamma_{e,l},\cdots,\gamma_{e,L} \}, 
\end{equation} 
where $\gamma_{e,r}$ is the instantaneous received SNR of each eavesdropper. Clearly, (\ref{gamma_Co_MRC}) corresponds to a maximum ratio combining (MRC) decoding which is the best strategy that the super eavesdropper can use. As we can see from ($\ref{gamma_Co_MRC}$), $\gamma_C$ is the sum of $L$ independent Fox's $H$-function distributed RVs, the PDF and CDF of $\gamma_C$ are thus respectively given by \cite[eqs. (8) and (9)]{8338118}
\begin{figure*}[t] 
\setcounter{MYtempeqncnt}{\value{equation}}
\setcounter{equation}{32}
\begin{subequations}
\begin{equation} \label{PDF_Colluding}
f_C(\gamma) = \frac{\eta_C}{\gamma} H \left[ \left. \hspace{-0.5ex} \begin{array}{*{20}c}
\left(\begin{matrix}
0,0\\0,1
\end{matrix} \right) \\
\left(\begin{matrix}
m_r,n_r +1\\p_r+1,q_r
\end{matrix} \right)_{r=1:L} \\
\end{array}  \hspace{-0.5ex} \right|\hspace{-0.5ex}
 \left.  {\begin{array}{*{20}c}
    {-}   \\
    {(1;(1)_{r=1:L})}  \\
    \left[\begin{matrix}
    {(1,1),(c_i+C_i,C_i)_{i=1:q_r}}  \\
    {(d_l+D_l,D_l)_{l=1:p_r}}   \\
    \end{matrix}\right]_{r=1:L}
\end{array}}   \right| \hspace{-0.5ex}
 {\begin{array}{*{20}c}
\left( \lambda_r \gamma \right)_{r=1:L}
\end{array}} \right],\quad \gamma > 0,
\end{equation}
\begin{equation} \label{CDF_Colluding}
F_C(\gamma) = \eta_CH \left[ \left. \begin{array}{*{20}c}
\left(\begin{matrix}
0,0\\0,1
\end{matrix} \right) \\
\left(\begin{matrix}
m_r,n_r +1\\p_r+1,q_r
\end{matrix} \right)_{r=1:L} \\
\end{array} \right|
 \left. {\begin{array}{*{20}c}
    {-}   \\
    {(0;(1)_{r=1:L})}  \\
    \left[\begin{matrix}
    {(1,1),(c_i+C_i,C_i)_{i=1:q_r}}  \\
    {(d_l+D_l,D_l)_{l=1:p_r}}   \\
     \end{matrix}\right]_{r=1:L}
\end{array}}\right|
 {\begin{array}{*{20}c}
\left( \lambda_r \gamma \right)_{r=1:L}
\end{array}} \right],
\end{equation}
\end{subequations}
\hrulefill
\vspace{-0.1cm}
\end{figure*}
where $\eta_C = \prod \limits_{e=1}^M \frac{\kappa_{E,e}}{\lambda_{E,e}}$.

Similarly, the PDF and CDF of instantaneous SNR deployed with SC scheme is given by \cite{752901}
\begin{subequations}
\begin{equation} \label{PDF_CoSC}
f_C(\gamma) = \sum \limits_{\tau = 1}^L f_{e,\tau}(\gamma)\prod \limits_{l=1,\\l\neq\tau}^L F_{e,l}(\gamma),
\end{equation}
\begin{equation} \label{CDF_CoSC}
F_C(\gamma) = \prod \limits_{l=1}^L F_{e,l}(\gamma),
\end{equation}
\end{subequations}
where $f_{e,\tau}(\gamma)$ and $F_{e,l}(\gamma)$ are the corresponding PDF and CDF of the instantaneous received SNR of each eavesdropper, which are given in terms of univariate Fox's $H$-function.

It is worthy to mention that the multivariate Fox's $H$-function PDF and CDF of the equivalent super-eavesdropper makes it difficult to seek the exact SOP and ASC for the colluding scenario. Resultantly, we intend to provide the lower bound of the SOP and exact PNZ for this case.
\subsection{Secrecy Characterization of SOP}
\begin{theorem}
The SOP over Fox's $H$-function wiretap fading channels in the presence of $L$-colluding eavesdroppers with MRC scheme is lower bounded by (\ref{SOP_MRC_Colluding}), shown at the top of next page.
\end{theorem}
\begin{figure*}[t] 
\setcounter{MYtempeqncnt}{\value{equation}}
\setcounter{equation}{34}
\begin{equation} \label{SOP_MRC_Colluding}
\mathcal{P}_{out,MRC}^L = \frac{\eta_C \kappa_B}{\lambda_B} H \left[ \left. \begin{array}{*{20}c}
\left(\begin{matrix}
n_0+1,m_0\\q_0+1,p_0+2
\end{matrix} \right) \\
\left(\begin{matrix}
m_r,n_r +1\\p_r+1,q_r
\end{matrix} \right)_{r=1:L} \\
\end{array}  \right|
 \left. {\begin{array}{*{20}c}
    {(1 - b_i- B_i,(B_i)_{r=1:L}),(1,(1)_{r=1:L})}  \\  
    {(0,(1)_{r=1:L}),(1-a_i-A_i,(A_i)_{r=1:L}),(1,(1)_{r=1:L})}   \\      
    \left[\begin{matrix}
    {(1,1),(c_i+C_i,C_i)_{i=1:q_2}}  \\
    {(d_l+D_l,D_l)_{l=1:p_2}}   \\
    \end{matrix}\right]_{r = 1:L}
\end{array}} \right|
 {\begin{array}{*{20}c}
 \left(\frac{\lambda_{e,r}}{\lambda_B R_s} \right)_{r = 1:L}
\end{array}} \right],
\end{equation}
\hrulefill
\vspace{-0.3cm}
\end{figure*}
\begin{proof}
Plugging (\ref{CDF_Non}) and (\ref{PDF_Colluding}) into $\mathcal{P}_{out,MRC} = \int_0^{\infty}  F_{B}(\gamma_0) f_{C}(\gamma_C) d\gamma_C$, then re-expressing the univariate Fox's $H$-function and multivariate Fox's $H$-function in terms of their definition, and performing the interchange of the Mellin-Barnes integrals and the definite integral, with the help of \cite[eqs.(3.194.3) and (8.384.1)]{gradshteyn2014table}, we arrive at the final expression of $\mathcal{P}_{out,MRC}$ in (\ref{SOP_MRC_Colluding}).
\end{proof}

\begin{theorem}\label{Theorem_6}
The SOP over Fox's $H$-function wiretap fading channels in the presence of $L$-colluding eavesdroppers with SC scheme is lower bounded by (\ref{SOP_SC_Colluding}), shown on the top of this page.
\end{theorem}
\begin{figure*}[t] 
\setcounter{MYtempeqncnt}{\value{equation}}
\setcounter{equation}{35}
\begin{equation} \label{SOP_SC_Colluding}
\mathcal{P}_{out,SC}^L =\sum \limits_{l=1}^L \frac{\eta_C \kappa_B}{\lambda_B} H \left[ \left. \begin{array}{*{20}c}
\left(\begin{matrix}
n_l,m_l\\q_l,p_l
\end{matrix} \right) \\
\left(\begin{matrix}
m_0,n_0 +1\\p_0+1,q_0+1
\end{matrix} \right) \\
\left(\begin{matrix}
m_r,n_r +1\\p_r+1,q_r
\end{matrix} \right)_{r =1:l-1} \\
\left(\begin{matrix}
m_{r+1},n_{r+1}+1\\p_{r+1}+1,q_{r+1}+1
\end{matrix} \right)_{r = l+1:L} \\
\end{array}  \right|
 \left. {\begin{array}{*{20}c}
    {(1 -d_i+D_i,(D_i)_{r =1:L} )}  \\  
    {(1-c_i+C_i,(C_i)_{r=1:L})}   \\
    {(1,1),(a_i+A_i,A_i)_{i=1:p_0}}  \\
    {(b_l+B_l,B_l)_{l=1:q_0}}   \\
    \left[\begin{matrix} 
    {(1,1),(c_i+C_i,C_i)_{i=1:p_r}}  \\
    {(d_l+D_l,D_l)_{l=1:q_r},(0,1)}   \\
    \end{matrix}\right]_{r=1:l-1}\\
    \left[\begin{matrix} 
    {(1,1),(c_i+C_i,C_i)_{i=1:p_r}}  \\
    {(d_l+D_l,D_l)_{l=1:q_r},(0,1)}   \\
    \end{matrix}\right]_{r=l+1,L}    
\end{array}} \right|
 {\begin{array}{*{20}c}
\frac{\lambda_B}{\lambda_{e,l}} \\  \left(\frac{\lambda_{e,r}}{\lambda_{e,l}}\right)_{r=1:l-1} \\ 
\left(\frac{\lambda_{e,r}}{\lambda_{e,l}}\right)_{r=l+1:L} \\ 
\end{array}} \right],
\end{equation}
\hrulefill
\vspace{-0.3cm}
\end{figure*}
\begin{proof}
Accordingly, by doing some simple substitutions, $\mathcal{P}_{out,SC}^L$ can be rewritten as
\begin{equation}
\mathcal{P}_{out,SC}^L = \sum \limits_{\tau = 1}^L \int_0^\infty F_B(R_s\gamma) f_{e,\tau}(\gamma)\prod \limits_{l=1,\\l\neq\tau}^L F_{e,l}(\gamma) d\gamma,
\end{equation}
next, by using the Mellin transform of multiple univariate Fox's $H$-function, the proof is achieved. 
\end{proof}
\subsection{Secrecy Characterization of PNZ}
\begin{theorem}
The PNZ over Fox's $H$-function wiretap fading channels in the presence of $L$-colluding eavesdroppers with MRC scheme is given by (\ref{PNZ_MRC_Colluding}), shown on the top of this page.
\end{theorem}
\begin{figure*}[t] 
\setcounter{MYtempeqncnt}{\value{equation}}
\setcounter{equation}{37}
\begin{equation} \label{PNZ_MRC_Colluding}
\mathcal{P}_{nz,MRC} = \frac{\eta_C \kappa_B}{\lambda_B} H \left[ \left. \begin{array}{*{20}c}
\left(\begin{matrix}
n_0,m_0\\q_0,p_0 + 1
\end{matrix} \right) \\
\left(\begin{matrix}
m_r,n_r +1\\p_r+1,q_r
\end{matrix} \right)_{r=1:L} \\
\end{array}  \right|
 \left. {\begin{array}{*{20}c}
    {(1-b_i -B_i;(B_i)_{r=1:L})_{i=1:q_0}}   \\
    {(1-a_i -A_i;(A_i)_{r=1:L})_{i=1:p_0},(0;(1)_{r=1:L})}  \\
    \left[\begin{matrix}    
    {(1,1),(c_i+C_i,C_i)_{i=1:q_r}}  \\
    {(d_l+D_l,D_l)_{l=1:p_r}}   \\
    \end{matrix}\right]_{r=1:L}
\end{array}} \right|
 {\begin{array}{*{20}c}
\left( \frac{\lambda_{e,r} }{\lambda_B} \right)_{r=1:L}
\end{array}} \right],
\end{equation}
\hrulefill
\vspace{-0.3cm}
\end{figure*}
\begin{proof}
Substituting (\ref{PDF_Non}) and (\ref{CDF_Colluding}) into (\ref{def_Pspc}), then re-expressing the multivariate Fox's $H$-function in terms of its definition and interchanging the order of two integrals, we directly obtain (\ref{PNZ_MRC_Colluding}).
\end{proof}
\begin{theorem}
The PNZ over Fox's $H$-function wiretap fading channels in the presence of $L$-colluding eavesdroppers with SC scheme is given by (\ref{PNZ_SC_Colluding}), shown on the top of this page.
\end{theorem}
\begin{figure*}[t] 
\setcounter{MYtempeqncnt}{\value{equation}}
\setcounter{equation}{38}
\begin{equation} \label{PNZ_SC_Colluding}
\mathcal{P}_{nz,SC} = \frac{\eta_C \kappa_B}{\lambda_B} H \left[ \left. \begin{array}{*{20}c}
\left(\begin{matrix}
n_0,m_0\\q_0,p_0 
\end{matrix} \right) \\
\left(\begin{matrix}
m_r,n_r +1\\p_r+1,q_r +1
\end{matrix} \right)_{r=1:L} \\
\end{array}  \right|
 \left. {\begin{array}{*{20}c}
    {(1-b_i -B_i;(B_i)_{r=1:L})_{i=1:q_0}}   \\
    {(1-a_i -A_i;(A_i)_{r=1:L})_{i=1:p_0}}  \\
    \left[\begin{matrix}    
    {(1,1),(c_i+C_i,C_i)_{i=1:q_r}}  \\
    {(d_l+D_l,D_l)_{l=1:p_r},(0,1)}   \\
    \end{matrix}\right]_{r=1:L}
\end{array}} \right|
 {\begin{array}{*{20}c}
\left(\frac{\lambda_{e,r} }{\lambda_B} \right)_{r=1:L}
\end{array}} \right],
\end{equation}
\hrulefill
\vspace{-0.3cm}
\end{figure*}
\begin{proof}
Substituting (\ref{PDF_Non}) and (\ref{CDF_CoSC}) into 
\begin{equation}
\begin{split}
\mathcal{P}_{NZ,SC} &= \int_0^\infty f_B(\gamma_B) F_{C,SC}(\gamma_B) d\gamma_B \\
& = \int_0^\infty f_B(\gamma_B) \prod\limits_{l=1}^L F_{e,l}(\gamma_B) d\gamma_B ,
\end{split}
\end{equation}
then following the same methodology used in Theorem \ref{Theorem_6}, the proof is obtained.
\end{proof}
\section{Numerical Results and discussions} \label{Sec_Num}
In this section, Monte-Carlo simulations are used to validate the analytical derivations obtained in Sections \ref{Ana_noncoll} and \ref{Label_Colluding}, particularly, over one special case of Fox's $H$-function wiretap fading channel, i.e., $\alpha-\mu$ wiretap fading channels\footnotemark[2]\footnotetext[2]{It is worthy to mention that (i) the $\alpha$-$\mu$ fading channel is implemented by using the WAFO toolbox\cite{WAFO}; (ii) the numerical evaluation of univariate and bivariate Fox's $H$-function for MATLAB implementations are based on the method proposed in \cite[Table. II]{peppas2012simple} and \cite[Appendix. A]{6141996}, respectively.}. It is noted that bullets represent the  simulation results whereas solid lines are used to show the analytical expressions.
\subsection{Non-colluding sceanrio}
In order to validate the analytical accuracy of our derivations, Monte-Carlo simulation outcomes together with analytical results are presented in Figs. \ref{Pout_fig}-\ref{ASC_Asy_fig}, with regard to the aforementioned three secrecy performance metrics over $\alpha-\mu$ fading channels. Apparently, these figures show that our mathematical representations are in perfect agreements with the simulation results.  
\begin{figure}[!t]
\centering
\includegraphics[width=\columnwidth]{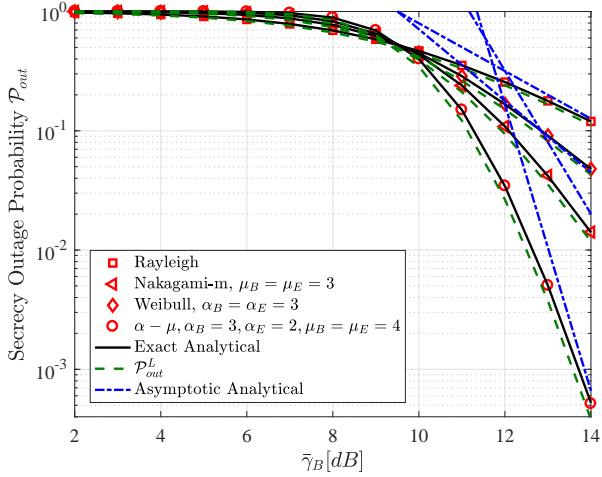}
\caption{$\mathcal{P}_{out}$ versus the average $\bar{\gamma}_B$ over Rayleigh, Nakagami-$m$, Weibull and $\alpha-\mu$ fading channels when $\bar{\gamma}_E=0$ dB and $R_t=0.5$, respectively.}
\label{Pout_fig}
\end{figure} 

In Fig. \ref{Pout_fig}, the SOP against $\bar{\gamma}_B$ is plotted for several fading scenarios, such as Rayleigh, Weibull, Nakagami-\textit{m}, and $\alpha-\mu$. As observed from the figure, specifically, the Nakagami-\textit{m} ($\alpha = 2$, $\mu =m$) against Rayleigh ($\alpha = 2$, $\mu=1$), and Rayleigh against Weibull ($\alpha$ is the fading parameter, $\mu =1$), one can conclude that larger $\alpha$ and $\mu$ values result in lower SOP. This is mainly because lower $\alpha$ and $\mu$ values represent serious non-linearity and sparse clustering, i.e., worse channel conditions \cite{7856980}. This phenomenon also remains true for the PNZ, as shown in Fig. \ref{Pnz_fig}. In addition, the lower bound of SOP and the asymptotic SOP are also plotted. It is observed that the lower bound of the SOP, i.e., $
\mathcal{P}_{out}^L$ offers a better SOP performance trend prediction, on the other hand, the asymptotic SOP gradually approximates the exact SOP with higher accuracy as $\bar{\gamma}_B$ increases. 
\begin{figure}[!t]
\centering
\includegraphics[width=\columnwidth]{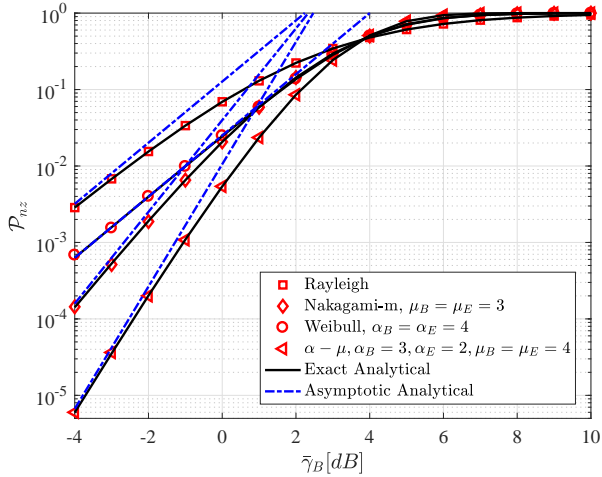}
\caption{$\mathcal{P}_{nz}$ versus the average $\bar{\gamma}_B$ for selected fading parameters when $\bar{\gamma}_E$ = 4 dB.}
\label{Pnz_fig}
\end{figure} 

As depicted in Fig. \ref{Pnz_fig}, both the exact and asymptotic behavior of $\mathcal{P}_{nz}$ are plotted against $\bar{\gamma}_B$ for Rayleigh, Weibull, Nakagami-\textit{m}, and $\alpha-\mu$. Compared with the exact result, one can conclude that our asymptotic PNZ behaves well at low $\bar{\gamma}_B$ regime. 
\begin{figure}[!t]
\centering
\includegraphics[width=\columnwidth]{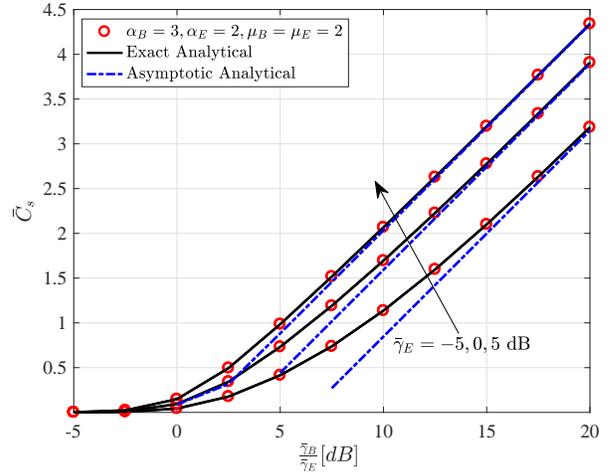}
\caption{$\bar{C}_s$ versus $\frac{\bar{\gamma}_B}{\bar{\gamma}_E}$ over $\alpha-\mu$ wiretap fading channels.}
\label{ASC_Asy_fig}
\end{figure} 

The ASC against the ratio of $\bar{\gamma}_B$ and $\bar{\gamma}_E$ is presented in Fig. \ref{ASC_Asy_fig}, and as expected, there is a perfect match between our analytical and simulated results. Also, one can obtain two insights from this graph: on one hand, lower $\alpha$ values lead to higher ASC, no matter whoever experiences severe fading. The insight obtained from this figure just vividly demonstrates how information-theoretic security exploits the fading property of wireless transmission medium to ensure secure transmission. On the other hand, a potential malicious eavesdropper can also benefit from poor channel conditions, since worse fading channels reversely enable them to better access and wiretap the main channel to a certain extent. Finally, to obtain a fair comparison, the asymptotic ASC is also depicted in Fig. \ref{ASC_Asy_fig}. Again, it can be seen that the asymptotic ASC presents a highly accurate approximation to the exact ASC, especially at high $\bar{\gamma}_B$ regime. 
\vspace{-0.3cm}
\subsection{Colluding scenario}
In this subsection, both the lower bound of SOP and PNZ are presented over $\alpha-\mu$, F-S $\mathcal{F}$, and EGK fading channels, respectively. For the simplicity of notations, it is assummed that all eavesdroppers undergo similar fading condition, i.e., similar fading parameters. It is noted that the implementation of multivariate Fox's $H$-function is available in Python \cite[Appendix A]{7089281} and MATLAB \cite{Hatim2018_K_Factor}. 

Fig. \ref{SOP_MRCSC_fig} demonstrates the analytical $\mathcal{P}_{MRC,out}^L$ and $\mathcal{P}_{SC,out}^L$ together with simulated SOP over $\alpha-\mu$, F-S $\mathcal{F}$, and EGK fading channels, respectively. One can perceive that our derived lower bound of SOP can closely approximate the exact SOP. As the number of cooperative eavesdroppers increases, the gap between the lower bound of SOP and exact SOP gradually becomes smaller. On the other hand, the increase of the number of $L$ contributes largely to the $\mathcal{P}_{out,MRC}^L$ when MRC scheme is employed, compared to the $\mathcal{P}_{out,SC}$ case.

Apart from Fig. \ref{SOP_AlphaMu_Fig}, we also compared the simulated and analytical SOPs for the following two scenarios: (i) changing $\gamma_E$ while fixing $R_t$, as shown in Fig.\ref{SOP_EGK_Fig}; and (ii) changing $R_t$ while keeping $\gamma_E$ constant, as depicted in Fig. \ref{SOP_FishF_Fig}. Apparently, one can obtain the following two observations. On one hand, Fig. \ref{SOP_EGK_Fig} shows that the lower bound of the SOP is becoming increasingly tight with the decrease of lower $R_t$. Different from \ref{SOP_EGK_Fig}, Fig.\ref{SOP_FishF_Fig} portrays that higher $\bar{\gamma}_E$ makes the lower bound of SOP sufficiently approximates the exact SOP. Those two observations can be mathematically explained from the definition of the lower bound of SOP, i.e., $\mathcal{P}(\gamma_B <(R_s \gamma_C + \mathcal{W})) \approx \mathcal{P}(\gamma_B <(R_s \gamma_C ))$. This condition can be met when $R_t$ goes to 0 ($\mathcal{W} = 2^{R_t} -1$), or $\gamma_C \gg \mathcal{W}$. 
\begin{figure} 
    \centering
  \subfloat[$\alpha-\mu$, $\alpha_B =2, \alpha_E =4, \mu_B = \mu_E = 3$\label{SOP_AlphaMu_Fig}]{%
        \includegraphics[width=\columnwidth]{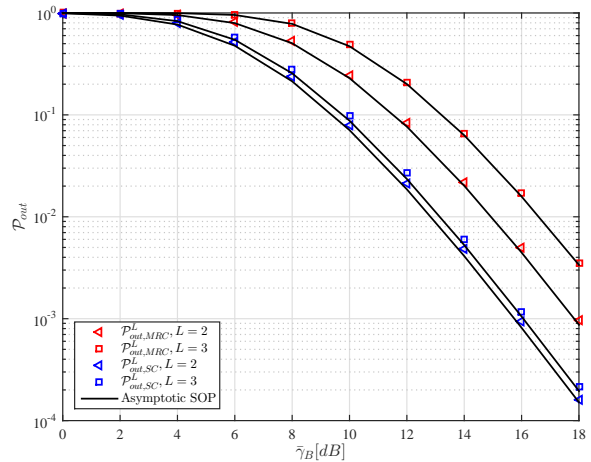}} \\
  \subfloat[EGK, $m_B =m_E= 2,m_{sB}=m_{sE}=4, \xi_B =\xi_{sB}=\xi_E =\xi_{sE}=1$\label{SOP_EGK_Fig}]{%
        \includegraphics[width=\columnwidth]{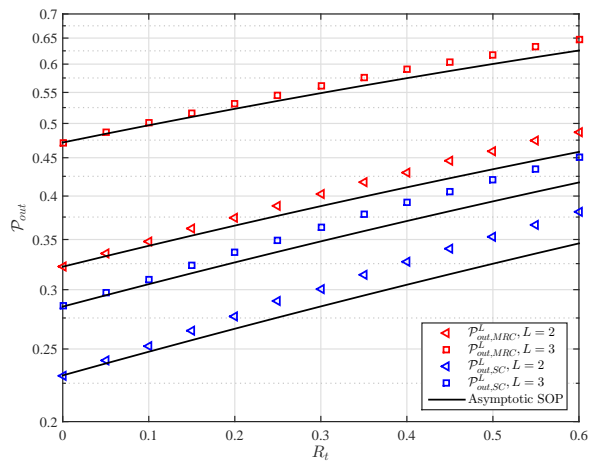}}\\    
  \subfloat[F-S $\mathcal{F}$, $m_B=m_E =2, m_{B,s}=m_{E,s} = 3$\label{SOP_FishF_Fig}]{%
        \includegraphics[width=\columnwidth]{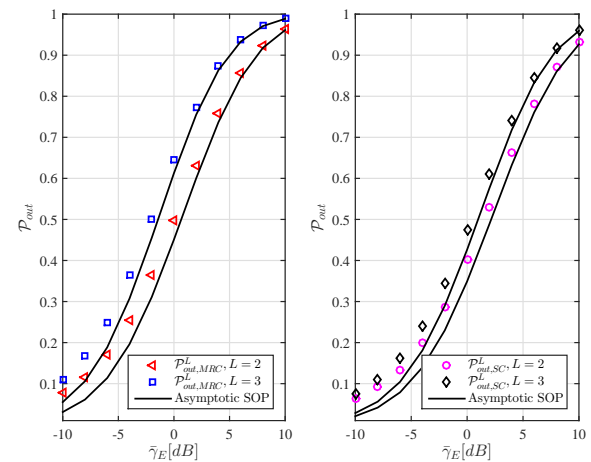}}\hfill
  \caption{The lower bound of SOP, i.e., $\mathcal{P}_{out}^L$ over Fox's $H$-function fading channels, (a) $\alpha- \mu$, $ \alpha_B= 2 $ (b) EGK; and (c) F-S $\mathcal{F}$.}
  \label{SOP_MRCSC_fig} 
\end{figure}

Likewise, in Fig. \ref{PNZ_MRCSC_fig}, the PNZ given in (\ref{PNZ_MRC_Colluding}) and (\ref{PNZ_SC_Colluding}) are plotted and compared with Monte-Carlo simulation. The validity of our presented PNZ expressions are examined over the $\alpha-\mu$, F-S $\mathcal{F}$, and EGK fading channels, respectively. Each figure witnesses perfect agreements between the exact analysis and simulated results. Besides, it is clear that the influences of $L$ on $\mathcal{P}_{nz,MRC}$ is larger than that on $\mathcal{P}_{nz,SC}$. This is obviously due to the MRC and SC schemes.
\begin{figure}[!t]
    \centering
  \subfloat[$\alpha-\mu, \alpha_B = 2, \alpha_{E} =4, \mu_B =  \mu_{E} = 3$\label{PNZ_AlphaMu_Fig}]{%
        \includegraphics[width=\columnwidth]{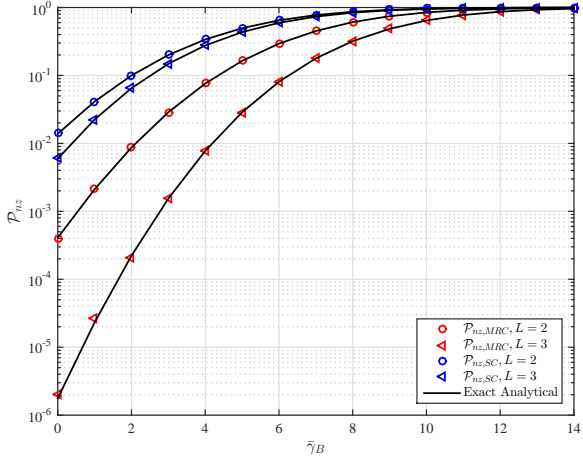}}\\
  \subfloat[EGK, $m_B =m_E= 2,m_{sB}=m_{sE}=4, \xi_B =\xi_{sB}=\xi_E =\xi_{sE}=1$\label{PNZ_EGK_Fig}]{%
        \includegraphics[width=\columnwidth]{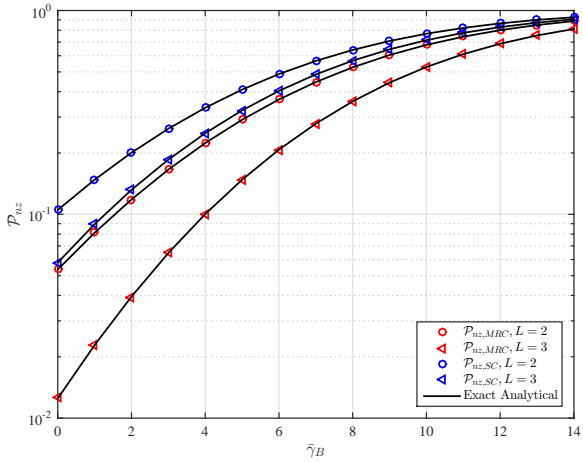}}\\    
  \subfloat[F-S $\mathcal{F}, m_B = m_E =2, m_{s,B} = m_{s,E} =3$\label{PNZ_FishF_Fig}]{%
        \includegraphics[width=\columnwidth]{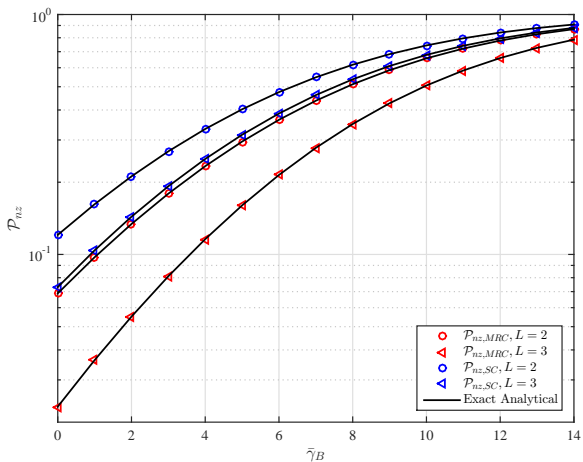}}\hfill
\caption{$\mathcal{P}_{nz,MRC}$, $\mathcal{P}_{nz,SC}$ versus $\bar{\gamma}_B$ over (a) $\alpha-\mu$, (b) F-S $\mathcal{F}$, and (c) EGK wiretap fading channels.}
\label{PNZ_MRCSC_fig}
\end{figure} 
\section{Conclusion} \label{Sec_con}
Since Fox's $H$-function fading channel can subsume most of the fading models, this paper comprehensively investigated the PLS over Fox's $H$-function wiretap fading channels, with consideration of the non-colluding and colluding eavesdropping scenarios. For the former non-colluding case, secrecy metrics, including the SOP, PNZ, and ASC, are derived with closed-form expressions in a general and unified manner. Those expressions are given in terms of the univariate or bivariate Fox's $H$-function. In addition, those closed-form expressions were further simplified to acquire the asymptotic behavior of the secrecy metrics. The asymptotic ones were much simpler and highly accurate for practical usage. In the presence of colluding eavesdroppers, a super eavesdropper employing by MRC or SC schemes were formulated, and subsequently the lower bound of SOP and the exact PNZ were provided in terms of multivariate Fox's H-function. Both scenarios are further demonstrated by Monte-Carlo simulations.

In addition, for the sake of providing more insights on some well-known fading models, several special cases of Fox's $H$-function distribution were particularly explored, including $\alpha-\mu$, F-S $\mathcal{F}$, and EGK. Those examples were further elaborated with the general form, and their accuracy was also compared with Monte-Carlo simulation results. As observed and discussed, the advantages of those general mathematical representations are listed as follows: (i) they are consistent with the existing works; (ii) they provide a unified generic approach to other fading models which can be expanded in terms of Fox's $H$-function fading distribution; and (iii) they provide a promising secrecy performance analysis framework when colluding eavesdroppers are undergoing different independent fading conditions.
\appendices
\section{Proof of the Theorem \ref{theorem_pout}} \label{appendix1}
At the very beginning, revisiting (\ref{Pout_theorem1})
\begin{equation} \label{Par_pout}
\begin{split}
 \mathcal{P}_{out} & = \int_0^{\infty}  F_{B}(\gamma_0) f_{E}(\gamma_E) d\gamma_E \\
 & = 1 - \int_0^{\infty}  \bar{F}_{B}(\gamma_0) f_{E}(\gamma_E) d\gamma_E \\
&=1 - \frac{1}{2\pi j}\int_{\mathcal{L}_1} \mathcal{M}[\bar{F}_B(\gamma_0),1-s] \mathcal{M}[f_E(\gamma_E),s] ds,
\end{split}
\end{equation}
 and using the definition of Mellin transform and Fox's $H$-function, we arrive at $\mathcal{M}[F_B(s)]$
\begin{equation} \label{Mell_1}
\begin{split}
&\mathcal{M}[F_B(\gamma_0),1-s] = \int_0^\infty \gamma_E^{-s} F_B(\gamma_0) d\gamma_E \\
& \mathop=^{(a)} \frac{\kappa_B}{2 \lambda_B\pi j} \int_{\mathcal{L}_1}\Theta_B^{F}(\xi) \lambda_B^{-\xi} \int_0^\infty \gamma_E ^{-s} \gamma_0^{-\xi} d\gamma_E d \xi ,
\end{split}
\end{equation}
where step $(a)$ is developed by interchanging the order of two integrals. The inner integral in (\ref{Mell_1}) can be further expressed as 
\begin{equation} \label{innerint}
\begin{split}
&\int_0^\infty \gamma_E ^{-s} \gamma_0^{-\xi} d\gamma_E  = \mathcal{W}^{-\xi}  \int_0^\infty \gamma_E ^{-s} \left(1 + \frac{R_s}{\mathcal{W}}\gamma_E  \right)^{-\xi} d\gamma_E  \\
& \mathop=^{(b)} \frac{\mathcal{B}(1-s,\xi + s -1)}{\mathcal{W}^{\xi}} \left( \frac{R_s}{\mathcal{W}} \right)^{s-1} \mathop=^{(c)}  \frac{\Gamma(1-s)\Gamma(\xi + s -1)}{\Gamma(\xi)\mathcal{W}^{\xi} \left( \frac{R_s}{\mathcal{W}} \right)^{1 -s} }  ,
\end{split}
\end{equation}
where step $(b)$ is developed from \cite[eq. (3.194.3)]{gradshteyn2014table}, and step $(c)$ is obtained by using $\mathcal{B}(x,y) = \frac{\Gamma(x)\Gamma(y)}{\Gamma(x+y)}$ \cite[eq. (8.384.1)]{gradshteyn2014table}.

Plugging (\ref{innerint}) into (\ref{Mell_1}), yields the result given in (\ref{Mell_Fb}),
\begin{figure*}[!t]
\setcounter{MYtempeqncnt}{\value{equation}}
\setcounter{equation}{33}
\begin{equation} \label{Mell_Fb}
\begin{split}
\mathcal{M}[F_B(\gamma_0),1-s] &\mathop=^{(d)} \frac{\kappa_B}{2\lambda_B \pi j}\left( \frac{R_s}{\mathcal{W}} \right)^{s-1} \Gamma(1-s)  \int_{\mathcal{L}_1}   \frac{\Gamma(\xi + s -1) \Theta_B^{F}(\xi)}{\Gamma(\xi)} (\lambda_B\mathcal{W})^{-\xi} d \xi \\
& = \frac{\kappa_B \Gamma(1-s)}{\lambda_B }\left( \frac{R_s}{\mathcal{W}} \right)^{s-1} 
 H_{p_0 + 2,q_0 +2}^{m_0 + 1,n_0+1} \left[ {  \lambda_B \mathcal{W} \left| \hspace{-1ex} {\begin{array}{*{20}c}
    {(1,1),(a_j + A_j,A_j)_{j=1:p_0},(0,1)}   \\
   {(s-1,1),(b_j + B_j,B_j)_{j=1:q_0},(0,1)}  \\
\end{array}} \right.} \hspace{-1.5ex} \right].
\end{split}
\end{equation}
\hrulefill
\vspace{-0.4cm}
\end{figure*}
and step $(d)$ is directly achieved from the definition of bivariate Fox's $H$-function.

Subsequently, substituting (\ref{Mell_Fb}) and $\mathcal{M}[f_E(\gamma_E),s] =\frac{\kappa_E \chi_E^f(s)}{\lambda_E^{s}}  $ into (\ref{Par_pout}), yields the following result
\begin{equation}
\begin{split}
&\mathcal{P}_{out} = 1-\frac{\kappa_B \kappa_E  \mathcal{W} }{4\lambda_B R_s \pi^2 } \int_{\mathcal{L}_1} \int_{\mathcal{L}_2} \frac{\Gamma(\xi + s -1) \Theta_B^{\bar{F}}(\xi)}{\Gamma(\xi) \left(\lambda_B\mathcal{W}\right)^{\xi}} \\
&\hspace{5ex}\times \Gamma(1-s) \Theta_E^f(s) \left( \frac{R_s}{\lambda_E\mathcal{W}} \right)^{s} d\xi ds,
\end{split}
\end{equation}
Next, deploying the definition of the bivariate Fox's $H$-function \cite{gradshteyn2014table}, the proof is achieved.

\section{Proof for Theorem \ref{theorem_asc} } \label{Appen_asc}
Since the logarithm function can be alternatively re-expressed in terms of Fox's $H$-function with the help from \cite[eq. (8.4.6.5)]{prudnikov1990integrals} and \cite[eq. (8.3.2.21)]{prudnikov1990integrals},
\begin{equation}
\ln(1+x) =  H_{2,2}^{1,2} \left[ {  x \left| {\begin{array}{*{20}c}
    {(1,1),(1,1)}   \\
   {(1,1),(0,1)}  \\
\end{array}} \right.} \right],
\end{equation}
For the ease of proof, we take the proof for $I_1$ as an example.
\begin{equation} \label{I1_Mell}
\begin{split}
I_1 &= \frac{1}{2\pi j} \int_{\mathcal{L}_1} \mathcal{M}[F_E(\gamma_B),s]\mathcal{M}[g(\gamma_B),1-s]ds,
\end{split}
\end{equation}
\begin{equation} \label{Mell_fB}
\begin{split}
\mathcal{M}[F_E(\gamma_B),s] = \frac{\kappa_E}{\lambda_E^{1+s}} \Theta_E^{F}(s),
\end{split}
\end{equation}
where $\Theta_E^{F}(s)$ is shown in (\ref{ThetaEs}) on the top of next page.
\begin{figure*}[!t]
\setcounter{MYtempeqncnt}{\value{equation}}
\setcounter{equation}{38}
\begin{equation}  \label{ThetaEs}
\begin{split}
\Theta_E^{F}(s)& = \frac{\Gamma(-s) \prod \limits_{l=1}^{m_1} \Gamma(d_l + D_l + D_l s) }{\Gamma(1-s)\prod \limits_{l=m_1 +1}^{q_1} \Gamma(1-d_l - D_l - D_l s) }  \frac{\prod \limits_{i=1}^{n_1}\Gamma(1 - c_i - C_i -C_i s)}{\prod \limits_{i=n_1+1}^{p_1}\Gamma( c_i + C_i + C_i s )},
\end{split}
\end{equation}
\hrulefill
\vspace{-0.4cm}
\end{figure*}

$\mathcal{M}[g_k(\gamma_k),1-s]$ can be regarded as the Mellin transform of the product of two Fox's $H$-function \cite[eq. (2.25.1.1)]{prudnikov1990integrals}, which is given by (\ref{Mell_gBB}) and shown on the top of next page.
\begin{figure*}[!t]
\setcounter{MYtempeqncnt}{\value{equation}}
\setcounter{equation}{39}
\begin{equation} \label{Mell_gBB}
\begin{split}
\mathcal{M}[g_B(\gamma_B),1-s]&=\int_0^\infty \gamma_B^{-s} \ln(1+\gamma_B) f_B(\gamma_B)  d\gamma_B = \kappa_B\int_0^\infty \gamma_B^{-s} H_{2,2}^{1,2} \left[ {  x \left| {\begin{array}{*{20}c}
    {(1,1),(1,1)}   \\
   {(1,1),(0,1)}  \\
\end{array}} \right.} \right] 
 H_{p_0,q_0}^{m_0,n_0} \left[ {  \lambda_B \gamma_B  \left| {\begin{array}{*{20}c}
    {(a_j,A_j)_{j=1:p_0}}   \\
   {(b_j,B_j)_{j=1:q_0}}  \\
\end{array}} \right.} \right] d\gamma_B \\
& =\frac{\kappa_B}{\lambda_B^{1-s}}H_{q_0+2,p_0+2}^{n_0+1,m_0 + 2} \left[ {  \frac{1}{\lambda_B}  \left|  {\begin{array}{*{20}c}
    {(1,1),(1,1),(1-b_j -(1-s)B_j,B_j)_{j=1:q_1}}   \\
   {(1,1),(1-a_j - (1-s)A_j,A_j)_{j=1:p_1},(0,1)}  \\
\end{array}} \right.}  \right]
\end{split}
\end{equation}
\hrulefill
\vspace{-0.4cm}
\end{figure*}

Next, substituting (\ref{Mell_fB}) and (\ref{Mell_gBB}) into (\ref{I1_Mell}), yields 
\begin{equation}
\begin{split}
I_1 = - \frac{\kappa_B \kappa_E}{4 \pi^2 \lambda_B \lambda_E } \int_{\mathcal{L}_1} \int_{\mathcal{L}_2} \frac{\Theta(s,\xi) \Theta(\xi) \Theta_E^F(s)}{ \left( \frac{\lambda_E }{\lambda_B} \right) ^{s}  \lambda_B^\xi} ds d\xi,
\end{split}
\end{equation}
where $\Theta(s,\xi)$ and $\Theta(\xi)$ are given by (\ref{Thetasxi}), shown on the top of next page.
\begin{figure*}[!t]
\setcounter{MYtempeqncnt}{\value{equation}}
\setcounter{equation}{41}
\begin{equation} \label{Thetasxi}
\begin{split}
\Theta(s,\xi)& = \frac{\prod \limits_{i=1}^{n_0} \Gamma(1-a_i - A_i + A_i s + A_i \xi) \prod \limits_{l=1}^{m_0} \Gamma(b_l + B_l -B_l s -B_l \xi) }{\prod \limits_{i=n_0+1}^{p_0} \Gamma(a_i + A_i - A_i s - A_i \xi) \prod \limits_{l=m_0 + 1}^{q_0} \Gamma(1-b_l - B_l + B_l s + B_l \xi)},\quad \Theta(\xi) = \frac{\Gamma(1+\xi)\Gamma(-\xi) \Gamma(-\xi)}{\Gamma(1-\xi)}.
\end{split}
\end{equation}
\hrulefill
\vspace{-0.4cm}
\end{figure*}

Next, replacing $\xi = -\eta$, $s = -\tau$, $I_1$ can be expressed as (\ref{I1_final}) in terms of the bivariate Fox's $H$-function. In particular, when $n_1$ = 0, $I_1$ is further simplified in terms of the extended generalized bivariate Fox's $H$-function. 
Following the same methodology, $I_2$ can be obtained. $I_3$ can be finally achieved from \cite[eq. (18)]{7089281}, 
\begin{equation}
\begin{split}
\hspace{-2ex}&I_3  = \frac{\kappa_E}{2\pi j} \int_{\mathcal{L}_1} \mathcal{M}\{ \ln (1 + \gamma_E),s \}  \lambda_E^{-1+s} \Theta_E^f(1-s) ds\\
\hspace{-2ex}& = \frac{\kappa_E}{\lambda_E} H_{q_1+2,p_1+2}^{n_1 + 1,m_1 + 2} \left[ { \frac{1}{ \lambda_E} \left| \hspace{-1ex} {\begin{array}{*{20}c}
    {(1,1),(1,1),(1-d_l-D_l,D_l)_{l=1:p_1}}   \\
   {(1,1),(1-c_i-C_i,C_i)_{i=1:q_1},(0,1)}  \\
\end{array}} \right.} \hspace{-1.5ex}\right].
\end{split}
\end{equation}
\section{Proof for Asymptotic ASC} \label{Appendix_ASC}
Specifically, at high $\bar{\gamma}_B$ regime, $I_1$ can be expanded at the pole, i.e., $\xi = 0$, since $\xi = 0$ is the second order pole, as such, by using the residue theorem, we have
\begin{equation}
\begin{split}
\rm{Res}\left[ \frac{\Theta(s,\xi) \Theta(\xi)}{\lambda_B^\xi} ,0\right] &= \lim_{\xi \rightarrow 0} \frac{d}{d\xi}  \frac{\xi^2 \Theta(s,\xi) \Gamma(\xi)^2 \Gamma(1 -\xi)}{\lambda_B^\xi\Gamma(1+\xi)} \\
\end{split}
\end{equation}
Using the fact that $\frac{d \Gamma(s)}{ds} = \Gamma(s)\Psi_0(s)$ and the general Leibniz rule, we have 
\begin{equation} 
\begin{split}
&\rm{Res}\left[ \frac{\Theta(s,\xi) \Theta(\xi)}{\lambda_B^\xi} ,0\right] = \Theta(s,0) \left[  \sum_{l=1}^{m_0} B_l \Psi_0(b_l + B_l + B_l s) \right. \\
& \left. -  \sum_{j=1}^{n_0} A_j \Psi_0(1 - a_j - A_j -A_j s) + \sum_{l=m_0 +1}^{q_0} B_l \Psi_0(1 - b_l - B_l - B_l s) \right. \\
& \left. -  \sum_{i=n_0 + 1}^{p_0} A_i \Psi_0(a_i + A_i + A_i s)   -\ln \lambda_B\right],
\end{split}
\end{equation}
and subsequently when $\dfrac{\lambda_E}{\lambda_B} \rightarrow \infty$, we evaluate the residue at $s$, where 
\begin{equation}
\hspace{-2ex} s = \max \left[0, \left(-\frac{b_l + B_l}{B_l} \right)_{l = 1,\cdots,m_0},  \left(\frac{c_i + C_i - 1}{c_i} \right)_{i = 1,\cdots,n_1}\right].
\end{equation}
Considering all poles are simple, we arrive at the derived asymptotic $I_1$.

Similarly, at high $\bar{\gamma}_B$ regime, $I_2$ can be obtained at the point $u = \min\left(\frac{b_l + B_l}{B_l}\right)_{l=1,\cdots,m_0}$, the proof is subsequently completed.
\section*{Acknowledgment}
This work has been supported by the ETS’ research chair of physical layer security in wireless networks. We thank the editor and anonymous reviewers for their time and constructive comments, which improve this paper. The authors would also like to thank Dr. Hongjiang Lei and Yousuf Abo Rahama for their valuable input.

\ifCLASSOPTIONcaptionsoff
  \newpage
\fi

\bibliographystyle{IEEEtran}
\bibliography{FoxHfadingChannel}

\end{document}